\documentclass[11pt,a4paper]{article}

\usepackage{times}
\usepackage{color}
\usepackage{amsmath,amstext,amsthm,amssymb,amsfonts}
\usepackage{a4wide}
\usepackage{comment}
\usepackage{float}
\usepackage{graphicx}
\usepackage{psfrag,subfig}
\usepackage{xspace}
\usepackage{algorithmic}
\usepackage{algorithm}
\usepackage{bbold}
\usepackage{enumerate}

\newcommand{\np}{{\em NP}\xspace}
\newcommand{\nphard}{\np-hard\xspace}

\newtheorem{theorem}{Theorem}[section]
\newtheorem{corollary}[theorem]{Corollary}
\newtheorem{lemma}[theorem]{Lemma}
\newtheorem{claim}[theorem]{Claim}
\newtheorem{clm}[theorem]{Claim}

{\theoremstyle{remark} \newtheorem{remark}[theorem]{Remark}}
{\theoremstyle{definition} \newtheorem{definition}[theorem]{Definition}}

\newenvironment{proofof}[1]{\begin{proof}[Proof of #1]}{\end{proof}}
\newenvironment{proofsketch}{\begin{proof}[Proof Sketch]}{\end{proof}}

\newenvironment{labellist}[1][A]
{\begin{list}{{#1}\arabic{enumi}.}{\usecounter{enumi}\addtolength{\leftmargin}{-1.5ex}}}
{\end{list}}

\newcommand{\mcc}{\mathcal{C}}
\newcommand{\mcs}{\mathcal{S}}
\newcommand{\mcp}{\mathcal{P}}
\newcommand{\mone}{\mathbb{1}}

\newcommand{\ef}{f}

\newcommand{\R}{\ensuremath{\mathbb R}}

\newcommand{\Hc}{\ensuremath{\mathcal{H}}}
\newcommand{\I}{\ensuremath{\mathcal I}}
\newcommand{\K}{\ensuremath{\mathcal K}}

\newcommand{\Pc}{\ensuremath{\mathcal P}}

\newcommand{\sm}{\ensuremath{\setminus}}
\newcommand{\es}{\ensuremath{\emptyset}}
\newcommand{\sse}{\subseteq}

\DeclareMathOperator{\poly}{poly}

\DeclareMathOperator{\argmin}{argmin}

\DeclareMathOperator{\aint}{int}
\DeclareMathOperator{\aext}{ext}

\newcommand{\e}{\ensuremath{\epsilon}}
\newcommand{\ve}{\ensuremath{\varepsilon}}

\newcommand{\vGm}{\ensuremath{\varGamma}}

\newcommand{\al}{\ensuremath{\alpha}}

\newcommand{\dt}{\ensuremath{\delta}}
\newcommand{\Dt}{\ensuremath{\Delta}}

\newcommand{\w}{\ensuremath{\omega}}

\newcommand{\htf}{\ensuremath{\hat f}}
\newcommand{\hl}{\ensuremath{\hat l}}
\newcommand{\ha}{\ensuremath{\hat a}}
\newcommand{\hb}{\ensuremath{\hat b}}
\newcommand{\htau}{\ensuremath{\hat \tau}}

\newcommand{\val}[1]{\ensuremath{|{#1}|}}
\newcommand{\tf}{\ensuremath{\tilde f}}
\newcommand{\tO}{\ensuremath{\tilde O}}
\newcommand{\diff}{\ensuremath{\Dt}}
\newcommand{\assign}{\ensuremath{\leftarrow}}
\newcommand{\poa}{PoA\xspace}

\title{Achieving Target Equilibria in Network Routing Games \\ without Knowing the Latency Functions}

\author{
  Umang Bhaskar\thanks{Dept. of Computing and Mathematical Sciences, California Institute of Technology. Work supported in part
  by a Linde/SISL postdoctoral fellowship and NSF grants CNS-0846025, CCF-1101470 and EPAS-1307794. Email: {\tt umang@caltech.edu}.}
\and
  Katrina Ligett\thanks{Dept. of Computing and Mathematical Sciences, California Institute of Technology. Work supported in part
  by the Charles Lee Powell Foundation and a Microsoft Faculty Fellowship. 
  Email: {\tt katrina@caltech.edu}.} 
\and
  Leonard J. Schulman\thanks{Dept. of Computing and Mathematical Sciences, California Institute of Technology. Work supported in part by NSF grants 1038578 and 1319745. Work
  performed in part at the Simons Institute for the Theory of Computing at UC
  Berkeley. Email: {\tt schulman@caltech.edu}.} 
\and
  Chaitanya Swamy\thanks{Dept. of Combinatorics and Optimization, University of Waterloo. Supported in part by NSERC grant
  32760-06, an NSERC Discovery Accelerator Supplement Award, and an Ontario Early
  Researcher Award.
  Email: {\tt cswamy@math.uwaterloo.ca}.}
}

\date{}

\begin{document}

\maketitle
\def\thepage{}
\begin{abstract}
The analysis of network routing games typically assumes, right at the onset, precise and detailed information about 
the latency functions. Such information may, however, be unavailable or difficult to obtain. Moreover, one is often 
primarily interested in enforcing a desired target flow as the equilibrium by suitably influencing player behavior in 
the routing game. We ask whether one can achieve target flows as equilibria {\em without knowing the underlying 
latency functions}. 

Our main result gives a crisp positive answer to this question. We show that, under fairly general settings, one can 
efficiently compute {\em edge tolls} that induce a given target multicommodity flow in a nonatomic routing game using 
a {\em polynomial number of queries} to an {\em oracle} that takes candidate tolls as input and returns the resulting 
equilibrium flow. This result is obtained via a novel application of the ellipsoid method. Our algorithm extends 
easily to many other settings, such as (i) when certain edges cannot be tolled or there is an upper bound on the total 
toll paid by a user, and (ii) general nonatomic congestion games. We obtain tighter bounds on the query complexity for 
series-parallel networks, and single-commodity routing games with linear latency functions, and complement these with 
a query-complexity lower bound. We also obtain strong positive results for Stackelberg routing to achieve target 
equilibria in series-parallel graphs.
 
Our results build upon various new techniques that we develop pertaining to the computation of, and connections 
between, different notions of approximate equilibrium; properties of multicommodity flows and tolls in series-parallel 
graphs; and sensitivity of equilibrium flow with respect to tolls. Our results demonstrate that one can indeed 
circumvent the potentially-onerous task of modeling latency functions, and yet obtain meaningful results for the 
underlying routing game. 
\end{abstract}

\newpage
\pagenumbering{arabic}
\normalsize

\section{Introduction} \label{intro}
{\em Network routing games} are a popular means of modeling settings where a collection of  
self-interested, uncoordinated users or agents route their traffic along an
underlying network---prominent examples include communication and transportation networks---%
and have been extensively studied from various perspectives in the Transportation
Science and Computer Science literature; 
see, e.g.,~\cite{Wardrop52,BeckmannMW56,Roughgarden05,Roughgarden08,FleischerJM04,KarakostasK04,YangH04,YangZ08,Swamy12},
and the references therein. 
These games are typically described in terms of an underlying directed graph $G=(V,E)$
modeling the network, a set of commodities specified by source-sink pairs and the volume
of traffic routed between them modeling the different user-types, and latency functions or
delay functions $(l^*_e:\R_+\mapsto\R_+)_{e\in E}$ on the edges, with $l^*_e(x)$ modeling the
delay experienced on edge $e$ when volume $x$ of traffic is routed along it.
The outcome of users' strategic behavior is described by the notion of an 
{\em equilibrium} traffic pattern, wherein no user may unilaterally deviate and reduce  
her total delay. 

The typical means of mathematically investigating network routing games 
takes the above specification as input, and thus, assumes, right at the onset, that
one has precise, detailed information about the underlying latency functions. 
However, 
such precise information may be unavailable or  
hard to obtain, especially in large systems, without engaging in 
a highly non-trivial and potentially-expensive modeling task. 
In fact, the task of capturing observed delays via suitable delay functions
is a topic of much research in itself in fields such as queuing theory and transportation
science.  
Recognizing that the modeling task of obtaining suitable latency functions is often really
a means to facilitating a mathematical analysis of the underlying routing game, we ask
whether one can sidestep this potentially-demanding task and analyze the routing game 
{\em without knowing the underlying latency functions}.  
This is the question that motivates our work.

In routing games, there is often a central authority who has some limited ability to
influence agents' behavior by making suitable changes to the routing game, 
e.g., imposing tolls on the network edges.  
This influence can be used to alleviate the detrimental effects of selfish agent behavior,
which might be expressed both in terms of the agents' costs (i.e., price of anarchy)
and externalities not captured by these (e.g., pollution costs in a road network).
Thus, a natural and well-studied goal in network routing games is to 
{\em induce a desirable target traffic pattern as an equilibrium} 
by suitably influencing agents' behavior. Such a target traffic pattern may be obtained by, e.g., limiting the traffic on every edge to a fraction of its capacity, or reducing the traffic near hospitals and schools.
It is evident here that suitably modeling the latency functions is only a means to the
end goal of achieving the target traffic pattern.
Our work aims to shed light on the following question: {\em can one achieve this end without the
means?}

\subsection{Our contributions}
We initiate a systematic study of network routing games from the perspective of achieving 
target equilibria without knowing the latency functions. 
We introduce a {\em query model} for network routing games to study such questions, and
obtain 
bounds on the query complexity of various tasks in this model.

\paragraph{The query model.}
We are explicitly given the underlying network $G=(V,E)$, the set of commodities specified
by the source-sink pairs and the demands to be routed between them, and the 
{\em target multicommodity flow} $f^*$ that we seek to achieve.
We {\em do not}, however, know the underlying latency functions $(l^*_e)_{e\in E}$. 
Instead, the only information that we can glean about the latency functions is via queries
to a {\em black box} or {\em oracle} (e.g., simulation procedure) that outputs the
equilibrium flow under a specified stimulus to the routing game. 
We investigate two methods for influencing agent behavior that have been
considered extensively in the literature,
which gives rise to two types of queries. 

We primarily focus on the task of computing {\em edge tolls} to induce $f^*$
(Sections~\ref{tolls} and~\ref{lb-tolls}). 
This yields the following query model:
each query consists of a vector of tolls on the edges, and returns the equilibrium flow
that results upon imposing these tolls.  
The goal is to minimize the number of queries required to compute tolls that yield $f^*$
as the equilibrium. 

We also explore, in Sections~\ref{stackelberg} and~\ref{lb-stackelberg}, the use
of {\em Stackelberg routing} to induce $f^*$. Here, we control an $\al$ fraction of 
the total traffic volume. 
Each query is a Stackelberg routing, which is a flow of volume at most $\al$ times the
total volume, and returns the equilibrium flow under this Stackelberg 
routing. 
The goal is to minimize the number of queries required to compute a Stackelberg routing  
\nolinebreak \mbox{that induces $f^*$ as the equilibrium.}

\paragraph{Our results and techniques.}
Our main result is a crisp and rather sweeping positive result showing that 
{\em one can always obtain tolls that induce a given target flow $f^*$ with a polynomial
number of queries} (Section~\ref{tolls-ellipsoid}).
With linear latency functions, 
our algorithm computes tolls that enforce $f^*$ exactly 
(Theorem~\ref{thm:ellipsoidnonatomic}). 
With more general latency functions, such as convex polynomial functions, equilibria
may be irrational, so it is not meaningful to assume that a query returns the exact
equilibrium. Instead, we assume that each query returns a (suitably-defined) approximate
equilibrium 
and obtain tolls that enforce a flow that is component-wise close to $f^*$
(Theorem~\ref{thm:ellipsoidapprox}).  

The chief technical novelty underlying these results is an unconventional application of the 
ellipsoid method. We view the problem as one where we are searching for the (parameters of
the) true latency functions $l^*$ and tolls that induce $f^*$. 
It is information-theoretically {\em impossible}, however, to identify $l^*$ (or even get 
close to it) in the query model since,---as is the case even when $G$ is a single
edge---there may be no way of distinguishing two sets of latency functions.
The key insight is that, notwithstanding this difficulty,  
if the current candidate tolls $\tau$ do not enforce $f^*$, then one can use the
resulting equilibrium flow 
to identify a hyperplane that separates our current candidate $(l,\tau)$ from the true
tuple $(l^*,\tau^*)$. This enables one to use the machinery of the ellipsoid method to
obtain tolls enforcing $f^*$ in a polynomial number of queries. 

Our ellipsoid-method based algorithm is quite versatile and can be easily
adapted to handle various generalizations (Section~\ref{tolls-extn}). 
For instance, we can incorporate {\em any} linear 
constraints that tolls inducing $f^*$ must satisfy, which one can separate over. This
captures constraints where we disallow tolls on certain edges, or place an upper bound
on the total toll paid by an agent.
All our machinery extends seamlessly to the more-general setting of 
{\em nonatomic congestion games}. 
Finally, another notable extension is to the setting of {\em atomic routing games}
under the assumption that the equilibrium is unique.  

In Sections~\ref{tolls-sepa} and~\ref{tolls-linear}, we 
devise algorithms with substantially improved query complexity for 
(a) multicommodity routing games on series-parallel (sepa) networks, and 
(b) single-commodity routing games on general networks, 
both with linear latency functions. 
For (a), we exploit the combinatorial structure of sepa graphs to design an algorithm with 
near-linear query complexity.  
We show that any toll-vector in a sepa graph can be converted
into a simpler canonical form, which can be equivalently viewed in terms of certain
labelings of the subgraphs of the sepa graph obtained via parallel joins; leveraging this
yields an algorithm with near-linear query complexity.
Our algorithm works more generally whenever we have an oracle that returns the (exact)
equilibrium. 
For (b), 
we prove that (roughly speaking) the equilibrium flow is a linear function of
tolls, and use linear algebra to infer the constants defining this linear map in
$\tO(|E|^2)$ queries. 

Complementing these upper bounds, we prove an $\Omega(|E|)$  {\em lower bound}
(Theorem~\ref{thm:tolllb}) on the query complexity of
computing tolls that induce a
target flow, even for single-commodity routing games on parallel-link graphs with linear
delays. 
This almost matches the query complexity of our algorithm for sepa graphs.  

En route to obtaining the above results, we prove various results that provide new insights
into network routing games even in the standard non-black-box model where latency
functions are known. For instance, we obtain results on: 
(a) the computation of approximate equilibria and their properties 
(Lemmas~\ref{e-eqbmcomp} and~\ref{eqbm-close}); 
(b) structural properties of tolls and multicommodity flows in sepa graphs
(Section~\ref{tolls-sepa}); and
(c) sensitivity of equilibrium flow with respect to tolls (Theorem~\ref{lem:linearcont}).
We believe that these results and the machinery we develop to obtain them are of
independent interest and likely to find various applications.

\medskip
In Section~\ref{stackelberg}, we investigate the use of Stackelberg routing to induce a
given target flow.
Stackelberg routing turns out to be significantly harder to leverage than edge tolls in
the query model. This is perhaps not surprising given that designing effective Stackelberg
routing strategies turns out to be a much-more difficult proposition than computing
suitable edge tolls, even in the standard non-black-box setting where latency functions 
are given (see, e.g.,~\cite{Roughgarden04,BonifaciHS10}). 
Nevertheless, we build upon the machinery that we develop for sepa graphs to give a rather 
efficient and general combinatorial algorithm that finds the desired Stackelberg routing
using at most $|E|$ queries to an oracle returning equilibrium flows. This applies to any
strictly increasing latency functions, and in particular, to linear latency functions. 
(Observe that this query complexity is even better than our query-complexity bound for 
inducing flows via tolls on sepa graphs.)
Moreover, our algorithm determines the Stackelberg routing of smallest volume that 
can induce $f^*$. 

We obtain various lower bounds in Section~\ref{lb-stackelberg} 
that allude to the difficulty of computing a Stackelberg routing in general networks that
induces a target flow. 
One possible strategy for finding such a  Stackelberg routing is to use the queries to
infer an (approximately) ``equivalent'' set of delay functions $l$, in the sense that any
Stackelberg routing yields the same (or almost the same) resulting equilibrium under the
two sets of delay functions. Then, since given the latency functions, it is easy to
compute a Stackelberg routing that induces a target flow (see Lemma~\ref{stack-basic}),
one can find the desired Stackelberg routing. 
Theorem~\ref{thm:stackelbergequivalence} shows that such an approach cannot work: 
in the query model, any algorithm that learns even an approximately equivalent set of
delay functions must make an {\em exponential} number of queries. 
Theorem~\ref{thm:stackelbergeq} proves an orthogonal computational lower bound showing
that determining the equivalence of two given sets of latency functions is an \nphard
problem. As in the case of tolls, along the way, we uncover a new result about the
hardness of Stackelberg routing. We show that the problem of finding a Stackelberg routing
that minimizes the average delay of the remaining equilibrium flow is \nphard to
approximate within a factor better than ${4}/{3}$ (Theorem~\ref{thm:eqinapprox}).
The query complexity of finding a Stackelberg routing in general networks that induces a 
target flow remains an interesting open question for further research. 

\medskip
Our results on tolls and Stackelberg routing 
demonstrate that 
it is indeed possible to circumvent the potentially-onerous task of modeling latency
functions, and yet obtain meaningful results for the underlying routing game. 
Our array of upper- and  lower- bound results indicate the richness of the query model,
and suggest a promising direction for further research.

\subsection{Related work}
Network routing/congestion games with nonatomic players---where each player controls an 
infinitesimal amount of traffic and there is a continuum of players---%
were first formally studied in the context of road traffic by Wardrop~\cite{Wardrop52},
and the equilibrium notion in such games is known as Wardrop equilibrium after him. 
Network routing games have since been widely studied in the fields of Transportation
Science, Operations Research, and Computer Science; see, e.g., 
the monographs~\cite{Roughgarden05,Roughgarden08} and the references therein. We limit
ourselves to a survey of the results relevant to our work.

Equilibria are known to exist in network routing games, even with atomic players with
splittable flow~\cite{Rosen65}. Nonatomic equilibria are known to be essentially unique,
but this is not the case for atomic splittable routing games, where uniqueness criteria
were recently obtained by Bhaskar et al.~\cite{BhaskarFHH09}. 
Equilibria in routing games are known to be inefficient, and considerable research in
algorithmic game theory has focused on quantifying this inefficiency 
in terms of the \emph{price of anarchy} (\poa)~\cite{KoutsoupiasP99,Papadimitriou01} of the
game, which measures, for a given objective, the worst-case ratio between the objective
values of an equilibrium and the optimal solution.   
A celebrated result of Roughgarden~\cite{Roughgarden03}, and Roughgarden and
Tardos~\cite{RoughgardenT02} gives tight bounds on the \poa for nonatomic routing games
for the social welfare objective. 
Recently, similar results were obtained for the \poa in atomic splittable routing
games~\cite{Harks11,RoughgardenS11}.

Given the inefficiency of equilibria, researchers have investigated ways of
influencing player behavior so as to alleviate this inefficiency.
The most common techniques studied to influence player behavior in network congestion
games are the imposition of tolls on the network edges, and Stackelberg routing. 
Network tolls are a classical means of congestion control, dating back to
Pigou~\cite{Pigou20}, and various results have demonstrated their effectiveness for both
nonatomic routing~\cite{BeckmannMW56,ColeDR03,FleischerJM04,KarakostasK04,YangH04}
and atomic splittable routing~\cite{Swamy12,YangZ08} showing that any minimal flow (in
particular, an optimal flow) can be enforced via suitable efficiently-computable tolls.  
Stackelberg routing has also been well studied, and it is known that this is much-less
effective in reducing the \poa. Whereas they can help in reducing the \poa to a constant
for certain network topologies such as parallel-link graphs~\cite{Roughgarden04} and
series-parallel graphs~\cite{Swamy12}, it is known that this is not possible for general
graphs~\cite{BonifaciHS10}. Furthermore, 
it is known that it is \nphard to compute the Stackelberg routing that minimizes the total
cost at equilibrium, even for parallel-link graphs with linear delay
functions~\cite{Roughgarden04}; a 
PTAS is known~\cite{KumarM02} for parallel-link graphs.  
All of these results 
pertain to the setting where one is given the latency functions.

To our knowledge, our query model has not been studied in the literature. 
It is useful to contrast our query model with 
work in {\em empirical game theory}, which also studies games when players'
costs are not explicitly given. 
In empirical game theory, each query specifies a (pure or mixed) strategy-profile, and returns the (expected) 
cost of each player under this strategy profile. In contrast, in our query model, 
we observe the equilibrium flow instead of individual player delays. 
This is more natural in the setting of routing games: in the absence of knowledge of 
the latency functions, one may only be able to calculate player delays under a strategy 
profile by routing players along the stipulated paths 
(and then observing player delays); but this may be infeasible since one cannot in fact  
impose routes on self-interested players.
Moreover, whereas our goal is to obtain a desirable outcome as the equilibrium,
the focus in empirical game theory is to compute an (approximate) equilibrium. 
Generic approaches to generate strategy-profiles for this purpose, and examples where
these have proved useful are discussed by Wellman~\cite{Wellman06}. An oblivious algorithm
that does not depend on player utilities, and instead uses best-responses to compute a
pure Nash equilibrium in bimatrix games was given by Sureka and Wurman~\cite{SurekaW05}.  
Starting with~\cite{PapadimitriouR08}, various papers have studied the complexity of
computing an exact or approximate correlated equilibrium in multi-player games using both
pure- and mixed-strategy queries~\cite{BabichenkoBP14,HartN13,JiangB13}.
More recently, Fearnley et al.~\cite{FearnleyGGS13} study algorithms in the
empirical-game-theory model for bimatrix games, congestion games, and graphical games,
and obtain various bounds on the number of queries required for equilibrium computation.

\section{Preliminaries and notation} \label{prelim}
A {\em nonatomic routing game} (or simply a routing game) is denoted by a tuple
$\varGamma = (G,l,\K)$, where $G=(V,E)$ is a directed graph with $m$ edges and $n$ nodes,
$l = (l_e)_{e \in E}$ is a vector of latency or delay functions on the edges of the graph,
and $\K=\{(s_i, t_i,d_i)\}_{i \le k}$ is a set of $k$ triples denoting sources, sinks, and
demands for $k$ {\em commodities}. The delay function 
$l_e:\mathbb{R}_+\mapsto\mathbb{R}_+$ gives the delay on edge $e$ as a 
function of the total flow on the edge. (Here, $\R_+$ is the set of nonnegative reals.) We 
assume that $l_e$ is continuous, and strictly increasing. 
For each commodity $i$, the demand $d_i$ specifies the volume of flow that is routed from
$s_i$ to $t_i$ by self-interested agents, each of whom controls an infinitesimal amount of
flow and selects an $s_i$-$t_i$ path as her strategy. The strategies selected by the
agents thus induce a multicommodity flow $(f^i)_{i \le k}$, where each 
$f^i=(f^i_e)_{e\in E}$ is an $s_i$-$t_i$ flow of value $d_i$. That is, the vector
$f^i=(f^i_e)_e$ satisfies:  
$$
f^i\geq 0, \qquad
\sum_{(v,w) \in E} f^i_{vw}-\sum_{(u,v) \in E} f^i_{uv}=0 
\quad \forall v\in V\sm\{s_i,t_i\}, \qquad
\sum_{(s,w) \in E} f^i_{sw}-\sum_{(u,s) \in E} f^i_{us}=d_i.
$$
We call $f=(f^i)_{i\leq k}$ a feasible flow. We say that $f$ is acyclic if 
$\{e: f^i_e>0\}$ is acyclic for every commodity $i$. 
We overload notation and use $f$ to also denote the total-flow vector $f=\sum_{i\leq k}f^i$. 
For a path $P$, we use $f_P>0$ to denote $f_e>0$ for all $e\in P$. 
We sometimes refer to $\bigcup_i\{s_i,t_i\}$ as the terminals of the routing game or
multicommodity flow. Given an $s$-$t$ flow $f$, we use $\val{f}$ to denote the value of  
$f$. 

Let $\mcp^i$ denote the collection of all $s_i$-$t_i$ paths. Given a multicommodity flow
$(f^i)_{i\leq k}$ induced by the agents' strategies, the  
delay of an agent that selects an $s_i$-$t_i$ path $P$ is the total delay, 
$l_P(f):=\sum_{e \in P} l_e(f_e)$, incurred on the edges of $P$.
Each agent in a routing game seeks to minimize her own delay. 
To analyze the resulting strategic behavior, we focus on the concept of a 
{\em Nash equilibrium}, which is a profile of agents' strategies where no individual
agent can reduce her delay by changing her strategy, assuming other agents do not change
their strategies. In routing games, this is formalized by the notion of 
{\em Wardrop equilibrium}.  

\begin{definition} 
A multicommodity flow $\htf$ is a {\em Wardrop equilibrium} (or simply an
equilibrium) of a routing game $\varGamma$ if it is feasible and for every commodity $i$,
and all paths $P$, $Q \in \mcp^i$ with $\htf_P^i>0$, we have $l_P(\htf) \le l_Q(\htf)$. 
A Wardrop equilibrium can be computed by solving the following convex program:
\begin{equation}
\min \ \Phi(f):=\sum_e \int_0^{f_e} l_e(x) \, dx \quad \text{s.t.} \quad 
f = \sum_{i=1}^k f^i, \quad f^i \text{ is an $s_i$-$t_i$ flow of value $d_i$} \ \
\forall i=1,\ldots,k. \label{eqn:wardrop}
\end{equation}
\end{definition}

Given a routing game $\varGamma$ and a feasible flow $f$, define 
$D^i(l,f):=\min_{P\in\Pc^i}l_P(f)$ 
for each commodity $i$, and 
call an edge $e$ a shortest-path edge for commodity $i$ with respect to $f$ 
if $e$ lies on some path $P\in\Pc^i$ such that $l_P(f)=D^i(l,f)$.
Let $\mcs^i(l,f)$ be the set of shortest-path edges for commodity $i$ with respect to
$f$.

\paragraph{Tolls, Stackelberg routing, and our query model.} 
We investigate both the use of edge tolls and Stackelberg routing to induce a given target
flow. Tolls are additional costs on the edges that are paid by every player that uses the
edge. A vector of tolls $\tau=(\tau_e)_e\in\R_+^E$ on the network 
edges thus changes the delay function on each edge $e$ to $l^\tau_e(x):=l_e(x)+\tau_e$,
and so the delay of an agent who chooses $P$ is now $l_P(f)+\tau(P)$, where 
$\tau(P):=\sum_{e\in P}\tau_e$. We use $f(l,\tau)$ to denote the equilibrium flow obtained
with delay functions $l=(l_e)_e$ and tolls $\tau=(\tau_e)_e$.
We say that $\tau$ enforces a multicommodity flow $f$ with latency functions $l$ if the
total flow $f(l,\tau)_e=f_e$ on every edge $e$.

For Stackelberg routing, in keeping with much of the literature, we focus on
single-commodity routing games. 
Given a single-commodity routing game $\vGm=(G,l,(s,t,d))$ and a parameter $\alpha \in [0,1]$,
a central authority 
controls at most an $\al$-fraction of the total $s$-$t$ flow-volume $d$ and routes this
flow in any desired way, and then the remaining traffic routes itself selfishly. 
That is, a Stackelberg routing $g$ is an $s$-$t$ flow of value at most $\al d$, which we
call the Stackelberg demand. 
The Stackelberg routing $g$ modifies the delay function on each edge
$e$ to $\tilde{l}_e(g; x):=l_e(x+g_e)$. The remaining $(1-\al)d$ volume of traffic routes
itself according to a Wardrop equilibrium, denoted by $f(l,g)$, of the instance 
$(G,\tilde{l}, (1-\alpha)d)$. The total flow induced by a Stackelberg routing $g$ is thus
$g+f(l,g)$. 

We shorten $f(l,\tau)$ to $f(\tau)$, and $f(l,g)$ to $f(g)$ when $l$ is clear from the
context.  

In our query model, we are given the graph $G$, the commodity set 
$\K=\{(s_i,t_i,d_i)\}_{i \le k}$, and a feasible {\em target multicommodity flow} $f^*$.  
There is an underlying routing game $\varGamma = (G,l^*,\K)$, to which we are
given query access. If our method of influencing equilibria is via tolls, then the oracle 
takes a toll-vector $\tau$ as input and returns the 
equilibrium flow $f(l^*,\tau)$ or a (suitably-defined) approximate equilibrium. 
Our goal is to minimize the number of queries required to compute tolls 
$\tau^*$ such that $f(l^*,\tau^*)=f^*$. 

If our method of influencing equilibria is via Stackelberg routing, then 
we are also given the parameter $\alpha\in[0,1]$. Each query
takes a Stackelberg routing $g$ with $\val{g}\leq \al d$ as input and returns 
the flow $f(l^*,g)$. Our goal is to minimize the number of queries required to compute a
Stackelberg routing $g^*$ of value at most $\al d$ such that $f(l^*,g^*)+g^* =f^*$, or determine
that no such Stackelberg routing exists.

\paragraph{Properties of equilibria.} 
The following facts about Wardrop equilibria, network tolls, and Stackelberg
routing will be useful. Recall that the delay functions are
nonnegative, continuous, and strictly increasing. 

\begin{list}{$\bullet$}{\itemsep=0.25ex \topsep=0.5ex \addtolength{\leftmargin}{-3ex}}
\item A feasible flow $f$ is an equilibrium flow iff 
$\sum_e (f_e-g_e)l_e(f_e)\leq 0$ for every feasible flow $g$; 
see, e.g.,~\cite{Roughgarden05}. 
Thus, the total-flow vector $(f_e)_e$ induced by an equilibrium flow is unique for 
strictly increasing delay functions.  

\item Every routing game admits an acyclic Wardrop equilibrium $\htf$. 
If the delay functions are polytime computable, then one can solve \eqref{eqn:wardrop} 
and compute: 
(i) $\htf$ in polytime for linear delay functions; 
(ii) an acyclic flow $f$ such that $\Phi(f)\leq\Phi(\htf)+\e$ 
in time $\poly\bigl(\text{input size},\log(\frac{1}{\e})\bigr)$. 
See, e.g.,~\cite{Roughgarden05}, for details. 

\item Every minimal feasible flow $f$ is enforceable via
tolls~\cite{FleischerJM04,KarakostasK04,YangH04}, where $f$ is minimal if 
there is no other feasible flow $g\neq f$ such that $g_e\leq f_e$ for every edge $e$. 
Given the edge delays $\bigl(l_e(f_e)\bigr)_e$, these tolls can be computed by solving an
LP, and are rational provided the commodity demands $(d_i)_i$ and the delays
$\bigl(l_e(f_e)\bigr)_e$ are rational.   

\end{list}

\noindent
The following lemma was essentially shown in~\cite{KaporisS09}; we include a
self-contained proof in Appendix~\ref{append-prelim}.

\begin{lemma} \label{stack-basic}
Let $(G,l,(d,s,t),\alpha)$ be a Stackelberg routing instance, and $f^*$ be a feasible
flow. 
Then, $f(g) + g = f^*$ for a Stackelberg routing $g$ iff $g_e \le f_e^*$ for every
edge $e$, and $g_e=f^*_e$ for all $e\not\in\mcs(l,f^*)$.  
\end{lemma}

\paragraph{Standard delay functions and encoding length.}
Our results hold for a broad class of underlying delay functions, that we now formally
describe. 
Throughout, we use $\I$ denote the input size of the given routing game.
We assume that we have an estimate $U$ with $\log U=\poly(\I)$ such that the target flow
$f^*$, the parameters of the unknown true delay functions $(l^*_e)_e$, and the
quantities that we seek to compute---tolls $\tau^*$ or the Stackelberg routing $g^*$
inducing $f^*$---all have encoding length $O(\log U)$. 
So we may assume that every $f^*_e$, $\tau^*_e$, $g^*_e$ value is a multiple of
$\frac{1}{U}$, and is at most $U$.  

When considering non-linear delay functions, we assume that the $l^*_e$s are convex
polynomials of degree at most some known constant $r$. Given the $O(\log U)$
encoding length, we may assume that all coefficients lie in $[0,U]$ and and are multiples
of $\frac{1}{U}$. 
We also assume that each $\frac{dl^*_e(x)}{dx}\geq\frac{1}{U}$ for all $x\geq 0$. 
We refer to such functions as {\em standard degree-$r$ polynomials}.
Under these conditions, it is easy to show (see Lemma~\ref{lat-props}) that there is
some constant $K:=K(r)=\poly(U,\sum_i d_i)$ such that every delay function $l^*_e$
satisfies 
\begin{alignat}{3}
(x-y)\bigl(l^*_e(x)-l^*_e(y)\bigr) & \leq\tfrac{\e^2}{K} 
\implies |x-y|\leq\e \quad && \text{for all $x,y,\e\geq 0$} 
\qquad  \label{invkcont} \\
|l^*_e(x)-l^*_e(y)| & \leq K|x-y| \qquad && \text{for all $x,y\in[0,{\textstyle \sum_i d_i}]$} 
\qquad  \label{klip} \\
l^*_e(2x) & \leq Kl^*_e(x) \qquad && \text{for all $x\geq 0$} 
\qquad  \label{kgrowth}
\end{alignat}

These properties are referred to as \emph{inverse-$K$-continuity}, \emph{$K$-Lipschitz}, and \emph{$K$-growth-boundedness} respectively.

\begin{lemma} \label{lat-props}
Let $l(x)=a_0+a_1x+\ldots+a_rx^r$ be a convex degree-$r$ polynomial such that
$a_1>0$, and all $a_i$s lie in $[0,U]$ and are multiples of $\frac{1}{U}$.
Then $l$ satisfies \eqref{invkcont}--\eqref{kgrowth} with
$K=\max\{U,2^r,rU(\sum_i d_i)^{r-1}\}$.  
\end{lemma}

\begin{proof}
Let $l'(x):=\frac{dl(x)}{dx}$ denote the derivative of $l$.
Since $l$ is convex, we have 
$|l(x)-l(y)|\geq |x-y|\cdot l'(\min\{x,y\})\geq|x-y|\cdot l'(0)\geq |x-y|/U$.
Therefore, $\frac{(x-y)^2}{U}\leq (x-y)\bigl(l(x)-l(y)\bigr)\leq\frac{\e^2}{K}$ and so
$|x-y|\leq\e$. 

Again, by convexity, $|l(x)-l(y)|\leq |x-y|\cdot l'(\max\{x,y\})$ and 
$l'(z)\leq rU(\sum_i d_i)^{r-1}\leq K$ for all $z\leq \sum_i d_i$.

Finally, it is clear that $l(2x)\leq 2^rl(x)\leq Kl(x)$ for all $x\geq 0$.
\end{proof}

\section{Inducing target flows via tolls} \label{tolls}
Recall that here we seek to compute tolls that enforce a given target flow $f^*$ given
black-box access to a routing game $\vGm^*=(G,l^*,(s_i,t_i,d_i)_{i\leq k})$, i.e.,
without knowing $l^*$.
Our main result is a crisp positive result showing that we
can always achieve this end with a polynomial number of queries by leveraging the
ellipsoid method in a novel fashion (Section~\ref{tolls-ellipsoid}).  
Our algorithm computes tolls that enforce: (a) $f^*$ exactly, for standard linear latency   
functions (where it is reasonable to assume that the black box returns the exact
equilibrium); and (b) a flow that is component-wise close to $f^*$, for standard
polynomial functions, where we now assume that each query only returns an approximate
equilibrium (see Definition~\ref{newe-eqbm}).
The main idea here is to view the parameters of the latency functions and
the tolls as variables, and use the ellipsoid method to search for the tuple 
$(l^*,\tau^*)$, where $\tau^*$ is such that $f(l^*,\tau^*)=f^*$. 
The key observation is that although we cannot hope to nail down $l^*$, 
given a candidate $(l,\tau)$ such that $f(l^*,\tau)\neq f^*$, one can derive a hyperplane
separating $(l,\tau)$ from $(l^*,\tau^*)$ using $f^*$ and the equilibrium flow
$f(l^*,\tau)$ returned by our oracle.

We showcase the versatility of our algorithm 
by showing that it is easily adapted to handle various extensions
(Section~\ref{tolls-extn}). 
For instance, we can impose {\em any} linear constraints on tolls given by a separation
oracle; examples include the constraint that certain edges cannot be tolled or that the
total toll paid by a user is at most a given budget.
Other notable extensions include the extension to general nonatomic congestion games, and 
to atomic splittable routing games under the assumption that the equilibrium is unique.

Finally, we devise algorithms with significantly improved query complexity for
multicommodity routing games on series-parallel (sepa) networks
(Section~\ref{tolls-sepa}), and single-commodity routing games on general networks
(Section~\ref{tolls-linear}), both with linear latency functions. 
We exploit the combinatorial structure of sepa graphs to design an algorithm with
near-linear query complexity, which almost matches the linear lower bound shown in
Theorem~\ref{thm:tolllb} for even parallel-link graphs with linear latencies.
For single-commodity routing games on general graphs with linear latencies, we show that
flows are linear functions of tolls and infer this linear map using $\tO(m^2)$ queries.

\subsection{An ellipsoid-method based algorithm for general routing
  games} \label{tolls-ellipsoid} 
The ellipsoid method for finding a feasible point starts by containing the
feasible region within a ball and generates a sequence of ellipsoids of successively
smaller volumes. In each iteration, one examines the center of the current ellipsoid. If
this is infeasible, then one uses a violated inequality to obtain a hyperplane, called a
separating hyperplane, to separate the current ellipsoid center from the feasible
region. 
One then generates a new ellipsoid by finding the minimum-volume ellipsoid containing the
half of the current ellipsoid that includes the feasible region.
We utilize the following well-known theorem about the ellipsoid method.

\begin{theorem}[\cite{ellipsoidbook}] \label{ellipthm}
Let $X\sse\R^n$ be a polytope described by constraints having encoding length at most
$M$. Suppose that for each $y\in\R^n$, we can determine if $y\notin X$ and if so, return
a hyperplane of encoding length at most $M$ separating $y$ from $X$. Then, we can use the
ellipsoid method to find a point $x\in X$ or determine that $X=\es$ in time
$\poly(n,M)$. 
\end{theorem}

\paragraph{Linear latencies.}
We first consider the case where each latency function $l^*_e(x)$ is a standard linear
function of the form $a^*_ex+b^*_e$, and our black box returns the exact
equilibrium flow induced by the input (rational) tolls. 
Thus, for every $e$, $a^*_e\in(0,U), b^*_e\in[0,U]$, and $a^*_e,b^*_e$ are multiples of
$\frac{1}{U}$.  
In a somewhat atypical use of the ellipsoid method, we use the ellipsoid method to search
for the point $(a^*_e,b^*_e,\tau^*_e)_e$. Abusing notation slightly, for a linear latency function $l(x)=ax+b$, we use $l$
to also denote the tuple $(a,b)$. 


\begin{theorem} \label{thm:ellipsoidnonatomic}
Given a target acyclic multicommodity flow $f^*$ and query access to
$\vGm^*$, we can compute tolls that enforce $f^*$ or determine that no such tolls exist,   
in polytime using a polynomial number of queries.    
\end{theorem}

\begin{proof}
We utilize the ellipsoid method and Theorem~\ref{ellipthm}. Given the center
$(\hl=(\ha_e,\hb_e)_e,\htau)$ of the current ellipsoid, we first check if
$\ha,\hb,\htau\geq 0$, and if not, use the violated constraint as the separating
hyperplane. Next, we use the black box to obtain
$g=f(l^*,\htau)$. If $g=f^*$, then we are done. Otherwise, we obtain a separating
hyperplane of encoding length $\poly(\I)$ as follows. 
(Note that the encoding length of $(\hl,\htau)$ is $\poly(\I)$.) 
We consider two cases.

\smallskip
\noindent {\bf Case 1: \boldmath $f(\hl,\htau)\neq f^*$.\ }
Note that we can determine this without having to compute the 
equilibrium flow $f(\hl,\htau)$. Since $f^*$ is acyclic, we can efficiently find a
commodity $i$, and $s_i$-$t_i$ paths $P, Q$ such that $f^*_P>0$ and
$\hl_P(f^*)+\htau(P)>\hl_Q(f^*)+\htau(Q)$. But since $f^*=f(l^*,\tau^*)$, we also have 
$l^*_P(f^*)+\tau^*(P)\leq l^*_Q(f^*)+\tau^*(Q)$. 
Thus, the inequality 
$$l_P(f^*)+\tau(P)\leq l_Q(f^*)+\tau(Q)$$ where the parameters of
$l$ and $\tau$ are variables yields the desired separating hyperplane.

\smallskip
\noindent {\bf Case 2: \boldmath $f(\hl,\htau)=f^*$.\ } 
Now since $g\neq f^*$ and is acyclic, we can again find
efficiently a commodity $i$ and paths $P, Q\in\Pc^i$ such that $g_P>0$ and 
$\hl_P(g)+\htau(P)>\hl_Q(g)+\htau(Q)$. Since $g=f(l^*,\htau)$, we also have 
$l^*_P(g)+\htau(P)\leq l^*_Q(g)+\htau(Q)$. Thus, the inequality 
$l_P(g)+\htau(P)\leq l_Q(g)+\htau(Q)$, where now {\em only} the $l_e$s are variables,
yields the desired separating hyperplane.
\end{proof}

\paragraph{Polynomial latency functions and approximate equilibria.}
We now consider the setting where the latency functions $(l^*_e)_e$ are standard
degree-$r$ polynomials, where $r$ is a known constant.
As before, 
we also use $l$ to denote the tuple of coefficients of the polynomial given by $l$.
Since the Wardrop equilibrium may now require irrational numbers,
it is unreasonable to assume that a query returns the equilibrium flow. 
So we assume that our black box returns an acyclic approximate equilibrium and show that
we can nevertheless compute tolls that induce an equilibrium that is component-wise close
to $f^*$.  
We first define approximate equilibria. 
Recall that $D^i(l,f)=\min_{P\in\Pc^i}l_P(f)$, and given tolls $\tau$, we define
$l^\tau_e(x):=l_e(x)+\tau_e$.

\begin{definition} \label{newe-eqbm} 
We say that a feasible flow $f$ is an {\em $\epsilon$-approximate equilibrium}, or simply 
an $\e$-equilibrium, of a routing game $(G,l,(s_i,t_i,d_i)_{i\leq k})$ if 
$\sum_e f_el_e(f_e)\leq \sum_i d_i\bigl(D^i(l,f)+\e\bigr)$. 
\end{definition}

Notice that our approximate-equilibrium notion is implied by the more-stringent (and
oft-cited) condition requiring that if $f_P>0$ for $P\in\Pc^i$ then 
$l_P(f)\leq D^i(l,f)+\e$. 
Importantly, our notion turns out to be weak enough that one can argue 
that an acyclic $\e$-equilibrium can be computed in time 
$\poly\bigl(\I,\log(\frac{1}{\e})\bigr)$ for any $\e>0$, 
{\em which lends credence to our assumption that the black box returns an acyclic
$\e$-equilibrium}, 
and yet is strong enough that one can leverage it within the framework
of the ellipsoid method (see Theorem~\ref{thm:ellipsoidapprox}).  
Unless otherwise stated, when we refer to a routing game below, we assume that the latency
functions satisfy the mild conditions \eqref{invkcont}--\eqref{kgrowth}, with $\log K$
being polynomially bounded. The following Lemma is proved in Appendix~\ref{append-tollsextn}.

\begin{lemma} \label{e-eqbmcomp}
Given a routing game with polytime-computable latency functions, one can compute
an acyclic $\e$-equilibrium in time $\poly\bigl(\I,\log(\frac{1}{\e})\bigr)$. 
\end{lemma} 

\begin{lemma} \label{eqbm-close}
Let $\htf$ be a Wardrop equilibrium and $g$ be an $\e$-equilibrium of a routing game
$(G,l,(s_i,t_i,d_i)_{i\leq k})$. 
Then, $\|g-\htf\|_\infty:=\max_e|g_e-\htf_e|\leq\sqrt{K\e\sum_i d_i}$.  
\end{lemma}

\begin{proof}
We have $\sum_e g_el_e(g_e)\leq\sum_{i} d_i\bigl(D^i(l,g)+\e\bigr)$ and 
$\sum_e \htf_el_e(g_e)\geq\sum_i d_iD^i(l,g)$. So 
$\sum_e (g_e-\htf_e)l_e(g_e)\leq\e\sum_i d_i$.
Also, $\sum_e (\htf_e-g_e)l_e(\htf_e)\leq 0$. So
$\sum_e (g_e-\htf_e)\bigl(l_e(g_e)-l_e(\htf_e)\bigr)\leq\e\sum_i d_i$. Each term of this
summation is nonnegative and hence, at most $\e\sum_i d_i$; therefore, 
$|g_e-\htf_e|\leq\sqrt{K\e\sum_i d_i}$ by inverse-$K$-continuity.
\end{proof}

Define an {\em $\e$-oracle for tolls} to be an oracle that receives tolls $\tau\in\R_+^E$
as input and returns an $\e$-equilibrium of the routing game
$(G,l^{*\tau},(s_i,t_i,d_i)_{i\leq k})$ having encoding length
$\poly\bigl(\I,\log(\frac{1}{\e})\bigr)$. 

\begin{theorem}
Let $f^*$ be a target acyclic multicommodity flow $f^*$ and $\dt>0$. 
Let $\e=\frac{\dt^2}{Kmk\sum_i d_i}$.  
Then, in time $\poly\bigl(\I,\log(\frac{1}{\dt})\bigr)$ and using
$\poly\bigl(\I,\log(\frac{1}{\dt})\bigr)$ $\e$-oracle queries, we can compute tolls $\tau$ 
such that $\|f(l^*,\tau) - f^*\|_\infty \le 2\dt$ or determine that no such tolls exist. 
\label{thm:ellipsoidapprox}
\end{theorem}

\begin{proof}
As before, we use the ellipsoid method. Let $(\hl,\htau)$ be the center of the current
ellipsoid. Assume that $\hl,\htau\geq 0$ and each function $\hl_e$ has
slope at least $\frac{1}{U}$; otherwise, we can use a violated constraint as the
separating hyperplane. 
We use the oracle with toll-vector $\htau$ to obtain an acyclic $\e$-equilibrium flow
$g$. Then, we have $\|g-f(l^*,\htau)\|_\infty\leq\sqrt{K\e\sum_i d_i}=\dt/\sqrt{mk}$ by
Lemma~\ref{eqbm-close}. 

We can efficiently determine if $f(\hl,\hat{\tau})\neq f^*$, and if so, then as in Case 1 in the
proof of Theorem~\ref{thm:ellipsoidnonatomic}, we can obtain a separating hyperplane of
encoding length $\poly\bigl(\I)$. So assume otherwise.

Now we check if $g$ is an $mk\e$-equilibrium for the latency functions $(\hl^{\htau}_e)_e$. 
If so, then $\|g-f^*\|_\infty\leq\dt$ and so $\|f(l^*,\htau)-f^*\|_\infty\leq 2\dt$ and we
are done.  
Otherwise, we find a valid path-decomposition $x=(x_{i,P})_{i,P\in\Pc^i}$ of $g$
having support of size at most $mk$. That is, we have $x\geq 0$, $\sum_{P\in\Pc^i}x_{i,P}=d_i$
for every commodity $i$, $\sum_i\sum_{P\in\Pc^i:e\in P}x_{i,P}=g_e$ for all $e$, and 
$\sum_i|\{P: x_{i,P}>0\}|\leq mk$. 
We may assume that every non-zero $x_{i,P}$ value has encoding length that is polynomial
in $\I$ and the size of $g$. Then 
$$
\sum_i\sum_{P\in\Pc^i}x_{i,P}\bigl(\hl^{\htau}_P(g)-D^i(\hl^{\htau},g)\bigr)=
\sum_e g_e\hl^{\htau}_e(g_e)-\sum_i d_iD^i(\hl^{\htau},g)>mk\e\sum_id_i
$$ 
where the last inequality follows since $g$ is not an $mk\e$-equilibrium for
$(\hl^{\htau}_e)_e$. 
Since the support of $x$ has size at most $mk$, this implies that there is some 
commodity $j$ and some path $R\in\Pc^j$ such that
$x_{j,R}\bigl(\hl^{\htau}_R(g)-D^j(\hl^{\htau},g)\bigr)>\e\sum_i d_i$. Moreover, we can find
such a $j$ and path $R\in\Pc^j$ efficiently by simply enumerating the paths in the support
of $x$. Let $Q\in\Pc^j$ be such that $\hl^{\htau}_Q(g)=D^j(\hl^{\htau},g)$. 

Since $g$ is an $\e$-equilibrium for the latency functions $(l^{*\htau}_e)_e$, again
considering the path-decomposition $x$, we have
$\sum_i\sum_{P\in\Pc^i}x_{i,P}\bigl(l^{*\htau}_P(g)-D^i(l^{*\htau},g)\bigr)\leq\e\sum_i d_i$.
Each term in this sum is nonnegative, so each term is at most 
$\e\sum_i d_i$. In particular, we have 
$x_{j,R}\bigl(l^{*\htau}_R(g)-l^{*\htau}_Q(g)\bigr)\leq
x_{j,R}\bigl(l^{*\htau}_R(g)-D^j(l^{*\htau},g)\bigr)\leq\e\sum_i d_i$.
So the inequality
$x_{j,R}\bigl(l_R(g)+\htau(R)-l_Q(g)-\htau(Q)\bigr)\leq\e\sum_i d_i$, 
with $l_e$s as the variables, is valid for $(l^*,\tau^*)$ but is violated by
$(\hl,\htau)$. 
This yields a separating hyperplane of encoding length
$\poly\bigl(\I,\log(\frac{1}{\e})\bigr)$.
\end{proof}

\subsection{Extensions} \label{tolls-extn}

\paragraph{Linear constraints on tolls given by a separation oracle.} 
Here, we require that the tolls $\tau^*$ imposing the target flow $f^*$ should lie in some 
polyhedron $X$, where $X$ is given by means of a separation oracle. This is rich enough to
model the following interesting scenarios.
\begin{list}{$\bullet$}{\topsep=0.25ex \itemsep=0ex \addtolength{\leftmargin}{-3ex}} 
\item A subset $F$ of edges cannot be tolled. This corresponds to the explicit
  constraint $\tau_e = 0$ for all $e\in F$.
\item The total toll paid by any player under the flow $f^*$ is at most a given budget
  $B$. This corresponds to the constraints $\tau(P)\leq B$ for every commodity $i$ and
  path $P\in\Pc^i$ with $f^{*i}_P>0$. One can separate over these exponentially-many
  constraints efficiently via a longest-path computation since $f^*$ is acyclic. 
\end{list}
The only change to our algorithm is that we first check if our current toll-vector $\htau$
lies in $X$. If not then the separation oracle provided yields the separating hyperplane;
otherwise, we proceed as before. The query complexity is now polynomial in the input size
and the encoding length of $X$. 

\paragraph{General nonatomic congestion games.}
This is a generalization of network routing games, where the graph is replaced by an
arbitrary set $E$ of resources, and $\Pc^i\sse 2^E$ is the strategy-set associated with
player-type $i$; a more complete definition appears in Appendix~\ref{append-tollsextn}.   
Our ellipsoid-based algorithm uses essentially no information about the underlying graph. 
We only require that given a congestion-vector $f$, we can find the maximum-delay set
$P\in\Pc^i$ for a given player-type $i$, and can find a valid decomposition of $f$ of
small support.
Both of these are trivial 
since the $\Pc^i$ sets are explicitly given in the input. Thus, our algorithms readily
extend to general nonatomic congestion games and 
Theorems~\ref{thm:ellipsoidnonatomic} and~\ref{thm:ellipsoidapprox} (with $mk$ 
replaced by $\sum_i|\Pc^i|$) continue to hold.

\paragraph{Atomic splittable routing games.}
Here, each commodity $i$ represents a {\em single} player who controls $d_i$ volume of
flow and her strategy is to choose an $s_i$-$t_i$ flow $f^i$ of value $d_i$. 
The cost incurred by a player $i$ under a feasible multicommodity flow (i.e., strategy 
profile) $f=(f^i)_{i\leq k}$ is $\sum_e f^i_el_e(f_e)$. 

Our results extend to atomic splittable routing games if 
we assume that for all valid choices of parameters of the latency functions and tolls
(as encountered during the ellipsoid method), the underlying atomic splittable
routing game has a unique Nash equilibrium. Here, by uniqueness we mean that if $f$ and
$g$ are two Nash equilibria, then $f^i_e=g^i_e$ for all commodities $i$ and edges $e$.
This is not without loss of generality, but is known to hold, for example, if all latency
functions are convex polynomials of degree at most 3, or if the graph is a generalized
nearly-parallel graph and $xl_e(x)$ is strictly convex for all $e$
(see~\cite{BhaskarFHH09}). 
When we say that tolls $\tau$ induce a flow $f^*=(f^{*i})_{i\leq k}$ here, we mean 
that the flow of every commodity $i$ on every edge $e$ is $f^{*i}_e$ in the resulting
equilibrium. 
%
Our result 
shows that the task of computing tolls that induce
specific {\em commodity-flows} 
can be reduced to the task of computing Nash equilibria (under the uniqueness assumption),
even in the black-box setting.  
Although, to our knowledge, no algorithm is known for either of these tasks, even when
latency functions are given, we believe that this reduction is of independent interest. 
The proof of Theorem~\ref{thm:atomic} is very similar to that of
Theorem~\ref{thm:ellipsoidnonatomic}: the only change is that to find the separating
hyperplane, we now consider the marginal delay functions instead of the delay functions;
see Appendix~\ref{append-tollsextn}. 

\begin{theorem} \label{thm:atomic}
In an atomic splittable routing game satisfying the aforementioned assumption, tolls that
induce a target flow $f^*=(f^{*i})_{i\leq k}$ at equilibrium, if they exist, can be
obtained with a polynomial number of queries to an oracle that returns the equilibrium
flow under tolls.  
\end{theorem}

\subsection{An algorithm for series-parallel networks with near-linear query complexity} 
\label{tolls-sepa}
We now give an algorithm for series-parallel networks with
$\tilde{O}(m)$ query complexity. This is a significant improvement over the
ellipsoid-based algorithm, 
and almost matches the linear lower bound proved in Theorem~\ref{thm:tolllb}
for single-commodity routing games on parallel-link graphs with linear latency functions.

\begin{theorem} \label{thm:sepatolls}  
On two-terminal series-parallel graphs, one can compute in polytime tolls that  
induce a given target multicommodity flow $f^*$ using $\tilde{O}(m)$ queries 
to an oracle that returns the equilibrium flow. 
Thus, we obtain $\tO(m)$ query complexity for multicommodity routing games with 
standard linear delay functions. 
\end{theorem}

We first recall some relevant details about series-parallel graphs.
%
A {\em two-terminal directed series-parallel graph}, abbreviated series-parallel (sepa)
graph, with terminals $s$ and $t$ is defined inductively as follows. 
A basic sepa graph is a directed edge $(s,t)$. 
Given two sepa graphs $G_1=(V_1, E_1)$ and $G_2=(V_2, E_2)$, with terminals $s_1$, $t_1$
and $s_2$, $t_2$ respectively, one can create a new sepa graph $G=(V,E)$ as follows.
A \emph{series join} of $G_1$ and $G_2$ yields the graph obtained by
identifying $t_1$ and $s_2$, with terminals $s = s_1$ and $t = t_2$. 
A \emph{parallel join} of $G_1$ and $G_2$ yields the graph obtained by
identifying $s_1$ and $s_2$, and $t_1$ and $t_2$; its terminals are $s = s_1 = s_2$ and
$t= t_1 = t_2$. 

For every series-parallel graph $G=(V,E)$, the recursive construction naturally yields a
binary \emph{decomposition tree}. The leafs of the tree are edges of $G$, and each
internal node specifies a series- or a parallel- join. Each node of the tree also
represents a subgraph of the $G$ (obtained by performing the joins specified by the
subtree rooted at that node), which is also clearly a sepa graph. 
In the sequel, we fix a decomposition tree corresponding to $G$.
Whenever we say a subgraph of $G$, we mean a subgraph corresponding to a node of this
decomposition tree.
Given a subgraph $H$, we use $s_H$, $t_H$ to denote its two terminals, and $\mcp(H)$ to
denote the set of all $s_H$-$t_H$ paths. We sometimes call $s_H$ and $t_H$, the source and
sink of $H$ respectively.
Let $\mathcal{H}$ be the collection of subgraphs corresponding to the parallel-join nodes
of the decomposition tree. 
For each $H\in \mathcal{H}$ obtained via the parallel join of $H_1$ and $H_2$, we identify
one of these as the ``left'' subgraph $H_L$ and the other as the ``right'' subgraph
$H_R$. Let $\Pc$ denote the set of all $s$-$t$ paths, where $s=s_G,\ t=t_G$. 
 
\paragraph{Proof outline.} 
Before we delve into the proof of Theorem~\ref{thm:sepatolls}, we give some intuition and
give a roadmap of the proof.
%
It is useful to first consider 
the simplest case of a graph with two parallel edges.
Observe that any target flow can be obtained by varying the \emph{difference} in tolls on
these two edges. Further, the correct difference in tolls can be obtained by a binary
search. Our key insight is that this intuition can be extended to series-parallel graphs
via a suitable transformation of tolls. We show that tolls required to obtain a
target flow can actually be described by the difference in tolls for each pair of parallel 
subgraphs, and then use binary search to obtain the correct differences that yield the
target flow. 

Formally, we show that any edge tolls in a sepa graph can in fact
be transformed into certain canonical tolls that are 
defined in terms of subgraphs (Claim~\ref{clm:tollearliest}). 
Further, formalizing the intuition that what is relevant is only the difference in tolls
on parallel subgraphs, we make the novel connection that canonical tolls are in fact
equivalent to labels on subgraphs $H\in\mathcal{H}$ (Lemma~\ref{clm:tollalphabeta}), 
where the label on subgraph $H\in\Hc$ stores the difference in
the canonical tolls of subgraphs $H_L$ and $H_R$ whose parallel-join yields $H$.

Thus, our problem reduces to finding the correct labels on subgraphs $H \in \mathcal{H}$,
which we aim to find via binary search. To do so, we establish certain
structural properties of multicommodity flows in sepa graphs (Lemma~\ref{lem:tollsepagood}).  
We leverage these to argue 
that if 
the canonical edge-tolls obtained from our current labels do not enforce the target flow,
then 
we can find a subgraph $H\in\Hc$ and deduce whether its label should be increased or
decreased.  
The query complexity is thus at most $|\Hc|$ times a logarithmic term
depending on the accuracy required and the parameters of the routing game. 
A detailed description appears after Claim~\ref{plusedges}.

The presence of multiple commodities complicates things, since in the particular decomposition tree that we fix for $G$, all edges in a subgraph may be 
shortest-path edges for one commodity but not for another. Thus creates problems with the
binary search since Claim~\ref{plusedges} may not hold.
We handle this by first arguing that there always exist tolls 
enforcing $f^*$ such that {\em every} $s$-$t$ path, and hence every $s_i$-$t_i$ path is a
shortest-path under edge costs $(l^{*\tau^*}_e(f^*_e))_e$ (Claim~\ref{clm:tollallequal}).

We believe that our structural insights into tolls and multicommodity flows on sepa
graphs 
are of independent interest and likely to find other applications. In fact, our
results on flows in sepa graphs also play an important role in our algorithm for inducing
target flows via Stackelberg routing in Section~\ref{stackelberg}.


\begin{clm}
For $\vGm^*=(G,l^*,(s_i,t_i,d_i)_{i\leq k})$ and target flow $f^*$ there exist tolls $\tau^*\in\R_+^E$  such that:

\begin{enumerate}[(i)]
\item $\min_{P\in\Pc}\tau^*(P)=0$; 
\item $l^{*}_P(f^*)+\tau^*(P)=l^{*}_Q(f^*)+\tau^*(Q)$ for every $i$ and paths
$P,Q\in\Pc^i$; and therefore 
\item $f(l^*,\tau^*)=f^*$.
\end{enumerate}
\label{clm:tollallequal}
\end{clm}

\begin{proof}
We will show that for any edge costs $(c_e)_e$, there exist tolls $\tau$ so that every $s$-$t$ path is a shortest path under edge costs $(c_e + \tau_e)_e$, and $\min_{P \in \mcp} \tau(P) = 0$. The claim follows simply by taking edge costs $(c_e=l^*_e(f^*_e))_e$ and
setting $\tau^*=\tau$, since every $s_i$-$t_i$ path clearly belongs to some $s$-$t$
path. 

The proof is by induction on the height of the decomposition tree for $G$. In the base
case, if the decomposition tree has height 1, $G$ consists of a single edge and setting
$\tau_e = 0$ satisfies the claim. For the inductive step, suppose $G$ is formed by the
composition of 
$H_1$ and $H_2$, and let $c^1$ and $c^2$ be the edge costs 
in subgraphs $H_1$ and $H_2$ respectively. Let $\tau^1$ and $\tau^2$ be the tolls that
satisfy the claim for costs $c^1$ in subgraph $H_1$, and costs $c^2$ in subgraph $H_2$
respectively. 

If $G$ consists of $H_1$ and $H_2$ composed in series, let $\tau_e = \tau_e^1$ if 
$e \in E(H_1)$ and $\tau_e = \tau_e^2$ otherwise. 
Then since any $s$-$t$ path $P$ consists of an $s_1$-$t_1$ path and
an $s_2$-$t_2$ path, each of which is a shortest path in $H_1$ and $H_2$ respectively,
every $s$-$t$ path is a shortest path. Secondly, by the inductive hypothesis, there is a
path $P$ in $H_1$ with $\tau^1(P) = 0$, and a path $Q$ in $H_2$ with $\tau^2(Q) = 0$. The
concatenation of paths $P$ and $Q$ yields an $s$-$t$ path $R$ with $\tau(R) = 0$. 

Suppose $G$ consists of $H_1$ and $H_2$ composed in parallel. 
For any paths $P \in \mcp(H_1)$ and $Q \in \mcp(H_2)$, let 
$\delta = c(Q) + \tau^2(Q)- c(P) - \tau^1(P)$. We may assume that $\dt\geq 0$
(otherwise switch $H_1$ and $H_2$).  
Note that by the inductive hypothesis the value of $\delta$ is independent of the choice
of $P$ and $Q$. Define tolls $\tau$ for graph $G$ as follows: 

\[
\tau_{vw} = \left \{ \begin{array}{ll}
	\tau^1_{vw} + \delta, & \mbox{ if $v=s$ and $(v,w) \in E(H_1)$.} \\
	\tau^1_{vw}, &\mbox{ if $v \neq s$ and $(v,w) \in E(H_1)$.} \\
	\tau^2_{vw}, & \mbox{ if $(v,w) \in E(H_2)$.}
	\end{array} \right.
\]

Then for any $s$-$t$ path $P$, if $P \in \mcp(H_1)$ then $c(P) + \tau(P)=c(P)
+ \tau^1(P) + \delta$. If $Q \in \mcp(H_2)$ then 
$c(Q) + \tau(Q)= c(Q) +\tau^2(Q)$. By definition of $\delta$ and the induction
hypothesis, every $s$-$t$ path is thus a shortest $s$-$t$ path. 
Since the tolls on paths in $H_2$ remain the same, there is also an $s$-$t$
path $R$ with $\tau(R) = 0$. 
\end{proof}

\begin{clm}
For any tolls $\tau\in\R_+^E$ on the edges of $G$, 
there exist $\alpha\in\R_+^E$ such that:
\begin{enumerate}[(i)]
\item $\tau(P) = \alpha(P)$ for all $P \in \mcp$, and
\item for every subgraph $H$ and every edge $e=(s_H, v) \in E(H)$, $\alpha_e \ge \min_{P \in \mcp(H)} \alpha(P)$.
\end{enumerate}
\label{clm:tollearliest}
\end{clm}

\begin{proof}
The proof is again by induction on the height of the decomposition tree. If $G$ is a
single edge, then $\alpha = \tau$. If $G$ is composed of subgraphs $H_1$ and $H_2$, let
$\tau^1$ and $\tau^2$ be the projection of $\tau$ onto the subgraphs. If $H_1$ and $H_2$
are in parallel, and tolls $\alpha^1$ and $\alpha^2$ satisfy the claim for tolls $\tau^1$
and $\tau^2$ in the subgraphs, it is easy to verify that tolls $\alpha$ defined by
$\alpha_e = \alpha_e^1$ for $e \in E(H_1)$ and $\alpha_e = \alpha_e^2$ for $e \in E(H_2)$
satisfy the claim. 

If $H_1$ and $H_2$ are in series, let $\alpha^1$ and $\alpha^2$ satisfy the claim for
tolls $\tau^1$ and $\tau^2$ in the subgraphs. Define $\delta = \min_{P \in \mcp(H_2)}
\alpha^2(P)$ and define the tolls 

\[
\alpha_{vw} = \left \{ \begin{array}{ll}
	\alpha^1_{vw} + \delta, & \mbox{ if $v=s_1$ and $(v,w) \in E(H_1)$} \\
	\alpha^1_{vw}, & \mbox{ if $v \neq s_1$ and $(v,w) \in E(H_1)$} \\
	\alpha^2_{vw} - \delta, &\mbox{ if $v= s_2$ and $(v,w) \in E(H_2)$} \\
	\alpha^2_{vw}, & \mbox{ if $v \neq s_2$ and $(v,w) \in E(H_2)$.}
	\end{array} \right.
\]

\noindent Any $s$-$t$ path $P$ consists of segment $P_1$ between vertices $s=s_1$ and
$t_1$, and segment $P_2$ between $t_1 = s_2$ and $t=t_2$. Then 

\[
\alpha(P) ~=~ \alpha(P_1) + \alpha(P_2) ~=~ \alpha^1(P_1) + \delta + \alpha^2(P_2) - \delta ~=~
\tau^1(P_1) + \tau^2(P_2) ~=~ \tau(P) \, . 
\]

\noindent Thus the first part of the claim is satisfied. 

For the second part, consider any subgraph $H$. If $H = G$, then since every path 
$P \in \mcp(H)$ consists of segments $P_1 \in \mcp(H_1)$ and $P_2 \in \mcp(H_2)$, for every edge
$e=(s,v) \in E$, 
\begin{align*}
\alpha_{sv} & = \alpha^1_{sv} + \delta \\
	& \ge \min_{P \in \mcp(H_1)} \alpha^1(P) + \min_{P \in \mcp(H_2)} \alpha^2(P) &
\mbox{ (by the inductive hypothesis and definition of $\delta$)}\\ 
	& = \min_{P \in \mcp(H_1)} \alpha^1(P) + \delta + \min_{P \in \mcp(H_2)} \alpha^2(P) - \delta \\
	& = \min_{P \in \mcp(H)} \alpha(P) \, .
\end{align*}

\noindent If $H \neq G$, then since every path path $P \in \mcp(H)$ contains exactly one
edge incident to $s_H$, the toll along every path changes by exactly the same quantity
($+\delta$, $-\delta$, or zero). 
\end{proof}

We call tolls $\al\in\R_+^E$ that satisfy property (ii) of Claim~\ref{clm:tollearliest}
{\em canonical tolls}. 
Thus, any edge tolls can be modified to obtain canonical edge tolls $\alpha$. 
These in turn can be mapped 
to a {\em labeling} $(L,\diff)$, where $\diff = (\diff_H)_{H \in \mathcal{H}} \in
\mathbb{R}_+^{\mathcal{H}}$, 
by setting $L=\min_{P \in \mcp} \alpha(P)$, 
and $\diff_H = \min_{P \in \mcp(H_L)} \alpha(P) - \min_{P \in \mcp(H_R)} \alpha(P)$ for
all $H\in\Hc$.
%
Lemma~\ref{clm:tollalphabeta} shows that this mapping is in fact invertible. Given the
labeling $(L,\diff)$ we can obtain canonical edge tolls $\al$ 
by the following procedure. Note that $|\Hc|\leq m$.

{\small \vspace{8pt} \hrule 
\begin{list}{M\arabic{enumi}.}{\usecounter{enumi} \topsep=0ex \itemsep=0ex
    \addtolength{\leftmargin}{-0.5ex}} 
\item Initialize $\al_e=0$ for all $e$.
\item We traverse subgraphs in $\Hc$ in a bottom-up manner, i.e., we consider all subgraphs 
in $\Hc$ that are descendants of $H\in\Hc$ before considering $H$. 
When we consider a subgraph $H$, we set $\al_e = \al_e+\max\{0,\diff_H\}$ for all
$e=(s_H,v)\in E(H_L)$, and  
$\al_e = \al_e+\max\{0,-\diff_H\}$ for all $e=(s_H,v)\in E(H_R)$.
\item Finally, we set $\al_e = \al_e + L$ for all $e = (s,v)\in E$.
\end{list}
\hrule}

\begin{lemma}
Let $(L,\diff)$ be the labeling obtained from some canonical tolls $\alpha\in\R_+^E$, 
and $\beta$ be the tolls obtained from $(L,\diff)$ by the above procedure. Then 
$\alpha=\beta$.  
\label{clm:tollalphabeta}
\end{lemma}

\begin{proof}
Let $\beta'$ be the tolls obtained after step 1 of the above procedure, i.e.,
before adding $L$ to the edges incident to $s$. We will show that for each edge $e$ not
incident to $s$, $\beta_e' = \alpha_e$, while for each edge $e$ incident to $s$, $\beta_e'
= \alpha_e - \min_{P \in \mcp} \alpha(P)$. 

The proof is by induction on the size of $G$. 
If $G = \{e\}$, then since there are no parallel compositions,
$\mathcal{H} = \emptyset$, and hence $\beta_e' = 0$ $ = \alpha_e - \min_{P \in \mcp}
\alpha(P)$. If $G$ is the series-join of $H_1$ and $H_2$, then for each
edge not incident to $s_{H_1}$ or $s_{H_2}$, $\beta_e' = \alpha_e$ by the inductive
hypothesis. Further, note that the minimum toll $\alpha(P)$ over all $s_{H_2}$-$t_{H_2}$
paths must be zero, since otherwise, on any edge $e=(s,v)$, $\alpha_e$ would be strictly
less than the minimum toll over $s$-$t$ paths. Hence by the inductive hypothesis
$\beta_e' = \alpha_e$ for edges that leave $s_{H_2}$. For edges incident to $s$, since any
$s$-$t$ path consists of a path between $s$ and $t_1 = s_2$ and between $s_2$ and $t$, and
by the inductive hypothesis, 
%
$$
\beta_e' = \alpha_e - \min_{P \in \mcp(H_1)} \alpha(P)  = \alpha_e - \min_{P \in \mcp}
\alpha(P)
$$
%
where the second equality follows because, as earlier observed, the minimum toll
$\alpha(P)$ over all $s_{H_2}$-$t_{H_2}$ paths must be zero. Thus the inductive hypothesis
holds in this case. 

If $G$ is the parallel-join of $H_1$ and $H_2$, then for each edge not
incident to $s$, $\beta_e' = \alpha_e$ by the inductive hypothesis. Further, assume
without loss of generality that $\min_{P \in \mcp(H_1)} \alpha(P) = \min_{P \in \mcp}
\alpha(P)$. Then by the inductive hypothesis, for each edge $e=(s,v) \in E(H_1)$, 
\[
\beta_e' = \alpha_e - \min_{P \in \mcp(H_1)} \alpha(P) = \alpha_e - \min_{P \in \mcp} \alpha(P)
\]
as stated in the claim. Let 
$\dt=\min_{P \in \mcp(H_2)}\alpha(P)-\min_{P \in \mcp(H_1)}\alpha(P)\geq 0$. 
By the procedure for computing $\beta'$, if $H_1=H_L$
and $H_2=H_R$, then $\diff_H = -\dt$, otherwise $\diff_H=\dt$. 
In both cases, when considering $G$, we only modify the tolls on edges of
$E(H_2)$ incident to $s_{H_2}$ by adding $\dt$ to these. So for each edge 
$e=(s,v)\in E(H_2)$, we have 
%
\begin{align*}
\beta_e' & = \alpha_e - \min_{P \in \mcp(H_2)} \alpha(P) + \dt \\
	& = \alpha_e - \min_{P \in \mcp(H_2)} \alpha(P) + \min_{P \in \mcp(H_2)} \alpha(P) - \min_{P \in \mcp(H_1)} \alpha(P) \\
	& = \alpha_e - \min_{P \in \mcp} \alpha(P)
\end{align*}
which completes the induction step, and hence, the proof.
\end{proof}

\begin{definition} \label{short:goodgraphs}
Given multicommodity flows $f$ and $\tf$, we call a pair $H_1$, $H_2$ of subgraphs, 
{\em $(f,\tf)$-discriminating} if:  
\begin{enumerate}[(i)]
\item the parallel-join of $H_1$ and $H_2$ is a subgraph in $\Hc$; 
and
\item $f_e > \tf_e$ for all $e \in E(H_1)$, and $f_e \le \tf_e$ for all $e \in E(H_2)$.
\end{enumerate}
\end{definition}

\begin{lemma} \label{lem:tollsepagood}
Let $f$ and $\tf$ be two feasible multicommodity flows for 
$(G,(s_i,t_i,d_i)_{i\leq k})$. If $f\neq\tf$, then there exists an
$(f,\tf)$-discriminating pair of subgraphs.
\end{lemma}

\begin{proofsketch}
We use induction on the series-parallel structure to first show a slightly weaker
statement: there exist subgraphs $H_1$ and $H_2$ whose parallel join is in $\Hc$ such that: 
(a) $f_e\geq\tf_e$ for all $e\in E(H_1)$, $f_e\leq\tf_e$ for all $e\in E(H_2)$, and 
(b) $\val{f_{H_1}}$, which we define to be the total flow routed under $f$ in $H_1$ for
commodities not internal to $H_1$, 
is greater than $\val{\tf_{H_1}}$, and $\val{f_{H_2}}<\val{\tf_{H_2}}$.
Now if $f_e>\tf_e$ for all $e\in E(H_1)$ then we are done. Otherwise, we show that if we
consider the minimal subgraph $K$ of $H_1$ (under the same decomposition tree used for $G$)
that contains both $f_e>\tf_e$ and $f_e=\tf_e$ edges, then $K$ must be a parallel-join
of subgraphs that form an $(f,\tf)$-discriminating pair.
\end{proofsketch}

We defer a full proof of Lemma~\ref{lem:tollsepagood} until Appendix~\ref{append-tollssepa}.

\begin{claim} \label{plusedges}
Let $\htf=f(l^*,\tau)$. If there is a subgraph $H$ such that $\htf_e>f^*_e$ for all 
$e\in E(H)$ then there is some commodity $i$ such that every $s_H$-$t_H$ path is 
part of a shortest $s_i$-$t_i$ path under edge costs $(l^{*\tau}_e(\htf_e))_e$.
\end{claim}

\begin{proof}
The proof is by induction on the size of $H$. If $H$ is an edge $e$, there is some
commodity $i$ such that $\htf^i_e>0$, so the statement holds. 
If $H$ is the parallel join of $H_1, H_2$, then it follows from the induction hypothesis
that every $s_{H}$-$t_{H}$ path must be of equal length (since there are commodities
corresponding to both $H_1$ and $H_2$); hence, there is a commodity corresponding to $H$  
and the statement follows.  
Suppose $H$ is the series composition of $H_1, H_2$. 
Let $\mathcal{K}$ be the set of commodities $i$ such that 
$\sum_{e=(s_H,v)\in E(H)}\htf^i_e>0$. 
For every $i\in\K$ such that $t_i\in V(H)\sm\{t_H\}$, the set of edges $(s_H,v)\in E(H)$
forms an $s_i$-$t_i$ cut, and so the flow across the cut must be the same in $\htf^i$ and
$f^{*i}$. However, $\sum_{e=(s_H,v)\in E(H)}\htf_e>\sum_{e=(s_H,v)\in E(H)}f^*_e$, so there
is some commodity $j\in\K$ such that $s_j,t_j\notin V(H) \setminus \{s_{H},t_{H}\}$. 
For commodity $j$, some $s_{H}$-$t_{H}$ path is part of a shortest $s_j$-$t_j$ path under
edge costs $(l^{*\tau}_e(\htf_e))_e$. Applying the induction hypothesis to $H_1, H_2$
yields that all $s_H$-$t_H$ paths are of the same length. Thus, every $s_{H}$-$t_{H}$ path
is part of a shortest $s_j$-$t_j$ paths under edge costs $(l^{*\tau}_e(\htf_e))_e$. 
\end{proof}

We now describe the algorithm for Theorem~\ref{thm:sepatolls}.
Let $\tau^*$ be tolls given by part (b) of Claim~\ref{clm:tollallequal} and $(0, \diff^*)$ be the labeling obtained from $\tau^*$. 
We may assume that $\tau^*_e\in[0,U']$ and is a multiple of $\frac{1}{U'}$ for all $e$, 
where $U'=m\poly(U,\sum_i d_i)$.
E.g., with standard linear latencies, since every $f^*_e, a^*_e, b^*_e \in [0,U]$ and is a multiple of 
$\frac{1}{U}$, we can take $U'=\max\{U^2,mK\sum_i d_i\}$.

{\small \vspace{7pt} \hrule \vspace{-1pt}
\begin{list}{T\arabic{enumi}.}{\usecounter{enumi} \topsep=0ex \itemsep=0pt
    \addtolength{\leftmargin}{-0.5ex}} 
\item Initialize, $L_H =-mU'$, $U_H = mU'$, $\diff_H=0$ for all $H \in \mathcal{H}$. 
Let $L=0$. Let $M=m\log(8mU'^2)$.
\item For $r=1,\ldots,M$, we do the following.
Map $(L,\diff)$ to canonical tolls $\al$ as described in steps M1--M3. 
Query the oracle to obtain $\htf=f(l^*,\al)$. If $\htf=f^*$, then exit the loop.
Otherwise, find an $(\htf,f^*)$-discriminating pair of subgraphs $H_1$, $H_2$ (which exists
by Lemma~\ref{lem:tollsepagood}). \\
Let $H$ be the parallel join of $H_1,H_2$. 
If $H_1 = H_L$, update $L_H\assign\diff_H$, else update $U_H\assign\diff_H$. 
If $|U_H-L_H|<\frac{1}{U'}$, set $\diff_H$ to be the multiple of $\frac{1}{U'}$ in
$[L_H,U_H]$; else update $\diff_H=(L_H+U_H)/2$. 
\item Return tolls $\al$.
\end{list}
\hrule }

\begin{proofof}{Theorem~\ref{thm:sepatolls}}
Let $\alpha^*$ be the canonical tolls obtained from $\tau^*$ via
Claim~\ref{clm:tollearliest},  
and let $(L^*,\diff^*)$ be the corresponding labeling. 
We have $L^* = 0$ due to Claims~\ref{clm:tollallequal} and~\ref{clm:tollearliest}.
The proof of Claim~\ref{clm:tollearliest} shows that, under the assumptions on
$\tau^*$, we have $\al^*_e$ is a multiple of $\frac{1}{U'}$, and is in $[0,mU']$ for all
$e$. Hence, $\diff^*_H\in[-mU',mU']$ and is a multiple of $\frac{1}{U'}$, for all
$H\in\Hc$. 

We say that the intervals $[L_H,U_H]$ assigned to $H\in\Hc$ are valid if
$\diff^*_H\in[L_H,U_H]$ for all $H\in\Hc$. We argue below that our algorithm maintains valid
intervals. Give this, in each iteration we halve the length of some interval, and this may
happen at most $\log (8mU'^2)$ times for the interval of some $H\in\Hc$ until we find
$\diff^*_H$, 
since $\diff^*_H$ is a multiple of $\frac{1}{U'}$. 
Since there are at most $m$ subgraphs in $\Hc$, after $M$ iterations (without reaching
$f^*$), we obtain $\diff^*$.  

We now prove that the algorithm maintains valid intervals.
Given tolls $\tau$ and a subgraph $H$, define $\tau_H:=\min_{P\in\Pc(H)}\tau(P)$.
So $\diff^*_H=\al^*_{H_L}-\al^*_{H_R}$.
The intervals are clearly valid at the start of the algorithm.
Suppose the intervals are valid at the start of an iteration in step T2. 
We may assume that $\htf\neq f^*$.
By Claim~\ref{plusedges}, there is some commodity
$i$ such that every $s_{H_1}$-$t_{H_1}$ path is part of a shortest $s_i$-$t_i$ path under
edge costs $(l^{*}_e(\htf_e)+\al_e)_e$. 
Let $P=\argmin_{P'\in\Pc(H_1)}\al^*(P')$ and $Q=\argmin_{Q'\in\Pc(H_2)}\al(Q)$.
Since $P$ is a segment of a shortest-path for commodity $i$, 
We have
$$
l^*_P(f^*)+\al_{H_1}
<l^*_P(\htf) + \al_{H_1}\le l^*_P(\htf) + \al(P)\le l^*_Q(\htf) + \al(Q)=l^*_Q(\htf) + \al_{H_2}
\leq l^*_Q(f^*)+\al_{H_2}.
$$
Here, the first and last inequalities follow since $H_1, H_2$ is
$(\htf,f^*)$-discriminating. The second inequality follows from the definition of
$\al_{H_1}$; the third, since $P$ is part of a shortest $s_i$-$t_i$ path; and the fourth
equality, from the definition of $Q$.
We know that every $s$-$t$ path is a shortest $s$-$t$ path under edge costs 
$(l^{*}_e(f^*_e)+\al^*_e)_e$. So we have
$$
l^*_P(f^*)+\al^*_{H_1}
= l^*_P(f^*) + \al^*(P) = l^*_Q(f^*) + \al^*(Q) 
\geq l^*_Q(f^*)+\al^*_{H_2}.
$$
Combining this with the earlier inequality gives 
$\al_{H_1}-\al_{H_2}<\al^*_{H_1}-\al^*_{H_2}$. So if $H_1=H_L$, then $\diff_H<\diff^*_H$;
otherwise, $\diff_H>\diff^*_H$. Thus, our update for $H$ preserves the validity of the
intervals.  
\end{proofof}

\begin{remark} \label{sgnoracle}
Our analysis shows that the above algorithm works 
whenever we have a ``sign oracle'' that given input tolls $\tau$ and a flow $f^*$, returns
the sign of $f(l^*,\tau)_e-f^*_e$ for all edges $e$. This is clearly weaker than having an
exact-equilibrium oracle. 
\end{remark}


\subsection{Nearly quadratic query complexity for single-commodity, linear-delay routing
  games} \label{tolls-linear}  

\begin{theorem}
For a single-commodity routing game $\varGamma$ with standard linear delay functions,
tolls that enforce $f^*$ can be obtained in at most $\tilde{O}(m^2)$ queries.  
\label{thm:linearmain}
\end{theorem}

Throughout, 
we assume without loss of generality that $f^* > 0$; otherwise, we impose infinite tolls
on any edge where $f^*_e = 0$, effectively removing these edges from the
graph.\footnote{The use of infinite tolls is a notational convenience; the same effect can
  be obtained with tolls $m^2 2^ml_{\max}(d)$.} We assume 
the delay function on any edge $e$ is $l_e(x) = a_e x + b_e$. Define $l_{\max}(x) :=
\max_{e \in E} a_e x + b_e$, and $\kappa(x) = x^2/Kd$. Define the \emph{support} of a flow
$f$ to be the set of edges with strictly positive flow. We will use negative tolls in our
proof; however, by Claim~\ref{clm:nonnegativetolls} which we prove in Appendix~\ref{append-tollslinear}, this is again just a notational
convenience. Similar arguments were used in~\cite{Fleischer05} to show boundedness of tolls, but the results are not directly applicable. Note that $f^*$ is acyclic.

\begin{clm}
For a single-commodity routing game and tolls $\tau$, there exist tolls $\tau' \ge 0$
so that $f(\tau) = \ef(\tau')$ and 
$\tau_{e'}' \le \tau_{e'} + \sum_{e:  \tau_e < 0} |\tau_e|$ for all $e'$. 
If the graph is acyclic, $\tau'$ can be obtained without knowledge of the delay functions.  
\label{clm:nonnegativetolls}
\end{clm}

\paragraph{Proof outline.} We show that if the support of the equilibrium flow remains
fixed, the equilibrium flow is a linear function of the tolls. Thus if we can obtain tolls
$\tau$ so that the support of $f(\tau)$ is the same as $f^*$, we can solve a linear system
of equations to obtain tolls that enforce $f^*$. Accordingly, our algorithm consists of
the following two steps. 

\smallskip

\noindent \textbf{Step 1: Enforcing the correct support.} We first obtain tolls $\tau$ so that $f_e(\tau) > 0 \Leftrightarrow f_e^* > 0$. By suitably large tolls on edges $e$ for which $f_e^* = 0$, we already have tolls that satisfy one direction of the implication. The other direction is roughly by binary search, described in Lemma~\ref{lem:linearsupport}: we pick an edge $r$ that does not yet have flow, and impose increasingly negative tolls on this edge until it has positive flow at the equilibrium. The difficulty here is in maintaining monotonicity of the support of the equilibrium flow. Increasing the flow on edge $r$ decreases flow on the other edges. We use a number of results regarding the sensitivity of equilibrium flow for this step. In fact, this step has quadratic query complexity, while the second step that actually obtains tolls that enforce $f^*$ has linear query complexity.

\smallskip

\noindent \textbf{Step 2: Obtaining the target flow $f^*$.} We now use Lemma~\ref{lem:lineartolls} which establishes the linearity of equilibrium flow as a function of tolls, if the support of the equilibrium flow does not change. Obtaining the coefficients of this linear map requires us to query the oracle with a small toll on each edge. The query complexity of this step is thus linear.

\smallskip

We start with some results about the continuity, monotonicity, and sensitivity of equilibrium flow as a function of tolls. Theorem~\ref{lem:linearcont}(ii) was earlier proved in~\cite{DafermosN84}.
Let $\mone_e\in\R^E$ be the vector with value 1 in coordinate $e$, and 0 everywhere else.

\begin{theorem} 
Let $\vGm$ be a single-commodity routing game with standard linear delay
functions. Then,
\begin{enumerate}[(i)]
\item $\|f(0)-f(\mone_e \kappa(\epsilon))\|_\infty \le \epsilon$,
\item $f(\tau)$ is continuous,
\item for edge $r$ and $\delta > 0$, $f_r(\mone_r \delta) \le f_r(0)$, and
\item for edge $r$ and $\delta > 0$,  $|f_r(-\mone_r \delta) - f_r(0)| \ge \|f(-\mone_r \delta) - f(0)\|_\infty$.
\end{enumerate}
\label{lem:linearcont}
\end{theorem}

The proof of (i) and (ii) are straightforward from the following immediate Corollary of Lemma~\ref{eqbm-close}. We prove (iii) and (iv) in Appendix~\ref{append-tollslinear}.

\begin{corollary}
For a multicommodity routing game $\varGamma$, let $\hat{f}$ be the equilibrium flow. If $g$ is a valid multicommodity flow that satisfies for all $i$, $P \in \mcp^i$, $g_P > 0 \Rightarrow l_P(g) \le D^i(l,g) + \epsilon$, then $\Vert g-\hat{f} \Vert_\infty \le \sqrt{K \epsilon \sum_i d_i}$.
\label{cor:indiffclose}
\end{corollary}

We now show a lemma that is used to prove the first step of our proof. Lemma~\ref{lem:linearsupport} shows that if edge $r$ has no flow or very little flow at equilibrium, then with a small number of queries we can obtain tolls so that the flow on edge $r$ increases, and the flow on the other edges does not change significantly.

\begin{lemma}
Let $\varGamma$ be a single-commodity nonatomic routing game, and let $\delta > 0$, $\delta \le d$. For tolls $\tau$, let $S := \{e: f_e(\tau) \ge \delta\}$, and edge $r \not \in S$. Then with $\log \left(-N/\kappa(\delta/3)\right)$ queries, we can determine tolls $\tau'$ so that $f_e(\tau') \ge \delta / 3$ for all $e \in S \cap \{r\}$, where
%
$N := \min_{P \in \mcp} \tau_P - \min_{P \in \mcp: r \in P} \tau_P - m l_{\max}(d) < 0\,$ .
\label{lem:linearsupport}
\end{lemma}

\begin{proof}
To obtain tolls $\tau'$, we will only vary the tolls on edge $r$. We thus parametrize tolls $\tau'$ by $\alpha$, where $\tau' = \tau + \mone_r \alpha$.

If $f_r(\tau) \ge \delta / 3$, we are done. Otherwise, we claim that if $\alpha = N$, then $f_r(\tau') = d$. To see this, let $Q$ be the path that minimizes $\sum_{e \in P} a_e d + b_e + \tau_e$ over all paths $P \in \mcp$ with $r \in P$, and let $f$ be the flow that sends the entire demand along this path. Then the delay along this path with tolls $\tau'$ is

\begin{align*}
\sum_{e \in Q} \left( a_e d + b_e + \tau_e'\right) & = \sum_{e \in Q} \left( a_e d + b_e + \tau_e\right) + N \\
	& = \min_{P: r \in P} \sum_{e \in P} \left( a_e d + b_e + \tau_e \right) + \min_{P \in \mcp} \tau(P) - \min_{P \in \mcp: r \in P} \tau(P) - m l_{\max}(d) \\
	& \le \min_{P \in \mcp} \tau(P) \, ,
\end{align*}

\noindent while for any path $P$ with $r \not \in P$, the delay along path $P$ is at least this quantity. Hence $f$ is actually an equilibrium flow, and if $\alpha \le N$, then $f_r(\tau') = d$.

Define $a$, $b \in [N,0]$ as follows.

\begin{align*}
a & := \max \{\alpha \in [N,0]: f_r(\tau') = \delta/3\} \\
b & := \min \{\alpha \in [N,0]: f_r(\tau') = 2\delta/3\} \\
\end{align*}

By the continuity of equilibrium flow with respect to tolls (Theorem~\ref{lem:linearcont},
(ii)), $a$, $b$ exist. By the monotonicity of equilibrium flow, for any $\alpha \in
[b,a]$, $f_r(\tau') \in [\delta/3, 2\delta/3]$. Then by Theorem~\ref{lem:linearcont}, (iv),
for any edge $e \in S$ and $\alpha \in [b,a]$, $f_e(\tau') \in [\delta/3, \delta]$. Thus
our problem reduces to finding an $\alpha \in [b,a]$, which we can find by binary
search. We will show that $a-b \ge \kappa(\delta/3)$, which gives us the bound on the
number of queries required. To see this, let $\tau^a$ and $\tau^b$ be the tolls obtained
by setting $\alpha = a$ and $\alpha = b$ respectively. Then if $a-b \le \kappa(\delta/3)$,
then by Theorem~\ref{lem:linearcont} (i), $\|f(\tau^a) - f(\tau^b)\| \le \delta/3$. 
\end{proof}

{\small \vspace{8pt} \hrule 
\begin{labellist}[F]
\item Initialize $\tau_e \leftarrow 0$ for all $e$, $i \leftarrow 1$, and $S \leftarrow \{e: f_e(\tau) \le d/3^i\}$.
\item While $S \neq E$
\item \qquad Pick an edge $r \not \in S$
\item \qquad By Lemma~\ref{lem:linearsupport}, find $\alpha \in [N,0]$ so that if $\tau' = \tau + \mone_r \alpha$, then $f_e(\tau') \ge d/3^{i+1}$ for all $e \in S \cup \{r\}$.
\item \qquad $\tau \leftarrow \tau'$, $i \leftarrow i+1$, $S \leftarrow \{e: f_e(\tau) \le d/3^i\}$
\end{labellist}
\hrule }

\begin{lemma}
The stated algorithm terminates with tolls $\tau$ so that $f_e(\tau) \ge d/3^m$ on every edge, and requires $O(m^2 \log (3m l_{\max}(d))$ queries.
\label{lem:fullsupport}
\end{lemma}

\begin{proofof}{Lemma~\ref{lem:fullsupport}}
Let $N(i)$ be the value of $N$ in the $i$th iteration of the while loop. Then by Lemma~\ref{lem:linearsupport}, the $i$th iteration requires $\log (-N(i)/\kappa(d/3^{i+1}))$ queries to complete, and adds at least one edge to the set $S$. Thus, there are at most $m$ iterations of the while loop. We will show that $|N(i)| \le m 2^{i-1} l_{\max}(d)$, thus proving the bound on the number of queries. Note that since all tolls are negative, $|\min_P \tau_P - \min_{P: r \in P} \tau_P|$ $\le |\min_P \tau_P|$.

The proof is by induction. In the first iteration since $\tau = 0$ initially, $N(1) \le m l_{\max}(d)$. In the $i$th iteration, there are at most $i-1$ other edges with tolls on them, and along any path the sum of the absolute values of these tolls is at most $\sum_{j \le i-1} 2^{j-1} m l_{\max}(d)$ $= (2^{i-1}-1) m l_{\max}(d)$, and hence $|N(i)| \le 2^{i-1} m l_{\max}(d)$.
\end{proofof}

This completes the first step of our proof. We now proceed with the second step. Lemma~\ref{lem:lineartolls} shows that the equilibrium flow is a linear function of the
tolls, as long as the set of edges with strictly positive flow remains constant. While a similar result on the linearity of the equilibrium flow was shown in~\cite{ColeDR06}, Lemma~\ref{lem:lineartolls} shows how to obtain the coefficients of the linear map.

\begin{lemma}
For any routing game $\varGamma$ and tolls $\tau^{(1)}$, let $f(\tau^{(1)}) > 0$. Then there exist coefficients
$(\beta_{e,e'})_{e,e' \in E}$ so that for any tolls $\tau$, 
\begin{enumerate}[(i)]
\item $f(\tau+\tau^{(1)}) > 0 \Rightarrow f(\tau+\tau^{(1)}) = f(\tau^{(1)}) + \beta \tau$, and
\item $f(\tau^{(1)}) + \beta \tau > 0 \Rightarrow f(\tau+\tau^{(1)}) = f(\tau^{(1)}) + \beta \tau$.
\end{enumerate}
\label{lem:lineartolls}
\end{lemma}

\begin{proof}
We first show how to obtain the coefficients $(\beta_{e,e'})_{e,e'}$. Define $f_{\min}$ $= \min_e f(0) > 0$. For each edge $e'$, 
let $\alpha^{e'} := \mone_{e'} \kappa(f_{\min}/2)$. By Corollary~\ref{cor:indiffclose}, $f(\tau^{(1)} + \alpha^{e'}) > f_{\min}/2$ for each edge $e'$. Then for
each edge $e \in E$, define $\beta_{e,e'} = \left( f_e(\tau^{(1)}+\alpha^{e'}) - f_e(\tau^{(1)}) \right) / \kappa(f_{\min}/2)$. 

Given tolls $\tau$, let $g := f(\tau^{(1)}) + \sum_{e'} \beta_{e,e'} \tau_{e'}$. In general, $g$
may be negative on some edges. However, we show that $g$ is an $s$-$t$ pseudoflow of value
$d$: it satisfies all the conditions for being a flow except nonnegativity. Further, we
show that $g$ is a minimizer of~(\ref{eqn:wardrop}) if we allow each $f$ to be a
pseudoflow, rather than a flow. 

To see the first claim, note that for a fixed edge $e'$ since $\beta_{e,e'}$ is the
difference of two (scaled) flows of the same value, it is a circulation. Then $g$ is the
sum of a flow and a set of circulations, and is hence a pseudoflow.  

To show that $g$ equalizes the delay on every $s$-$t$ path with tolls $\tau$, for any
$s$-$t$ path $p$, 

\begin{align}
\sum_{e \in p} l_e(g) - l_e(f(\tau^{1})) & = \sum_{e \in p} a_e \left( g_e - f_e(\tau^{1}) \right) = \sum_{e \in p} a_e \sum_{e'} \beta_{e,e'} \tau_{e'}  \nonumber \\
	& = \sum_{e'} \frac{\tau_{e'}}{\kappa(f_{\min}/2)} \sum_{e \in p} a_e \left( f_e(\tau^{(1)} + \alpha^{e'}) - f_e(\tau^{(1)}) \right) \label{eqn:gpathequal}
\end{align}

\noindent Since $f(\tau^{(1)} + \alpha^{e'})$ and $f(\tau^{(1)})$ are equilibrium flows with tolls $\tau^{(1)} + \alpha^{e'}$
and $\tau^{(1)}$ respectively, and both are strictly positive on every edge, it follows
from~(\ref{eqn:gpathequal}) that 

\begin{align*}
\sum_{e \in p} l_e(g) - l_e(f(\tau^{(1)})) & = \sum_{e'} \frac{\tau_{e'}}{\kappa(f_{\min}/2)} \left( D(f(\tau^{(1)}+\alpha^{e'})) - D(f(\tau^{(1)}))- \sum_{e \in p} \alpha_e^{e'} \right) \nonumber
\end{align*}

\noindent and since $\alpha_e^{e'} = 0$ for $e \neq e'$,

\begin{align*}
\sum_{e \in p} l_e(g) - l_e(f(\tau^{(1)})) & = \sum_{e'} \frac{\tau_{e'}}{\kappa(f_{\min}/2)} \left( D(f(\tau^{(1)}+\alpha^{e'})) - D(f(\tau^{(1)})) \right) - \sum_{e \in p} \tau_e \, .\nonumber
\end{align*}

\noindent Thus for any path $p$, 

\begin{align*}
\sum_{e \in p} l_e(g) + \tau_e & = \sum_{e \in p} l_e(f(\tau^{(1)}))  + \sum_{e'} \frac{\tau_{e'}}{\kappa(f_{\min}/2)} \left( D(f(\tau^{(1)}+\alpha^{e'})) - D(f(\tau^{(1)})) \right) \, .
\end{align*}

\noindent Further, for any path $p$, $\sum_{e \in p} l_e(f(\tau^{(1)})) = D(f(\tau^{(1)} - \sum_{e \in p} \tau^{(1)}_e$. Hence for any path $p$, $\sum_{e \in p} l_e(g) + \tau_e + \tau^{(1)}_e$ is equal. It follows immediately that if the second condition in the lemma is true, i.e., if $g >
0$, then $g$ must be an equilibrium flow with tolls $\tau^{(1)} + \tau$, and since the equilibrium is
unique, $f(\tau + \tau^{(1)}) = g$. This completes the proof of the second statement. 

For the first statement, for $0 \le \lambda \le 1$ define $h(\lambda) = f(\tau^{(1)} + \tau) +
\lambda(g -f(\tau^{(1)} + \tau))$. Since on any path $p$ as shown earlier $\sum_{e \in p} l_e(g) + \tau_e + \tau_e^{(1)}$ is equal, and $f(\tau^{(1)} + \tau) > 0$ by assumption,
this is also true for $h(\lambda)$. Further since $f(\tau^{(1)} + \tau) > 0$, there exists $\lambda >
0$ so that $h(\lambda) > 0$. Then $h(\lambda)$ must also be an equilibrium flow with tolls
$\tau$. By the uniqueness of equilibria, this is only possible if $f(\tau + \tau^{(1)}) = g$. 
\end{proof}

{\small \vspace{8pt} \hrule 
\begin{labellist}[L]
\item Use the earlier algorithm to get tolls $\tau^{(1)}$ so that $f(\tau^{(1)}) \ge d/2^m$.
\item Obtain the coefficients $(\beta_{e,e'})_{e,e'}$ as in Lemma~\ref{lem:lineartolls}
\item Solve the linear equations $\beta \tau^{(2)} = f^* - f(\tau^{(1)})$ for tolls $\tau^{(2)}$. Then $f(\tau^{(2)} + \tau^{(1)}) = f^*$.
\end{labellist}
\hrule }

\begin{proofof}{Theorem~\ref{thm:linearmain}.} We will show that the algorithm is correct, and requires $O(m^2 \log (3m l_{\max}(d))$ queries. The correctness of the first step follows from Lemma~\ref{lem:fullsupport}. To use Lemma~\ref{lem:lineartolls}, since $f_{\min} \ge d/{2^m}$, to obtain the coefficients $(\beta_{e,e'})_{e,e'}$, we require an additional $m$ queries, each of which applies an additional toll (relative to $\tau^{(1)}$) of $\kappa(d/2^{m+1})$ on individual edges.

Let $\tau^*$ be tolls such that $f(\tau^*) = f^*$. By the
first part of Lemma~\ref{lem:lineartolls}, then $f^* = f(\tau^{(1)}) + \beta (\tau^* - \tau^{(1)})$. Now $\tau^{(2)}$ is a solution to the system of linear equalities $\beta \tau^{(2)} = f^* - f(\tau^{(1)})$; since
$\tau^* - \tau^{(1)}$ satisfies this, we know a solution exists. Further, by the second part of the
Lemma since $f(\tau^{(1)}) + \beta \tau^{(2)} = f^* > 0$, in fact $f(\tau^{(1)} + \tau^{(2)}) = f^*$. Hence $\tau^{(1)} + \tau^{(2)}$ are the
tolls required to obtain the target flow. 
\end{proofof}


\section{Inducing target flows via Stackelberg routing on series-parallel graphs} 
\label{stackelberg} 
Recall that here we have a single-commodity routing game $\vGm^*=(G,l^*,(s,t,d))$. 
We are given a parameter $\al\in[0,1]$ and a target flow $f^*$, and we seek an $s$-$t$
flow $g$ of value of at most $\al d$ such that $g+f(l^*,g)=f^*$, if one exists.
We abbreviate $f(l^*,g)$ to $f(g)$.
We consider the setting where $G$ is a directed sepa graph with terminals $s$ and $t$, and 
devise an efficient algorithm 
that computes a Stackelberg routing inducing $f^*$ using at most $m$ queries to an oracle
that returns the equilibrium flow. 
The flow $g$ we compute is in fact of minimum value among  
all Stackelberg flows that induce $f^*$. (So either $g$ is the desired Stackelberg flow,
or none exists if $\val{g}>\al d$.)
Our algorithm works for arbitrary increasing delay functions provided, as in
Section~\ref{tolls-sepa}, we have an oracle that returns the correct sign of
$((f(g)+g\bigr)_e-f^*)_e$ given a Stackelberg routing $g$. In particular, the algorithm
works for increasing linear latencies.

As before, we fix a decomposition tree for $G$, and a subgraph
refers to a subgraph corresponding to a node of this tree. 
For a flow $f$ and subgraph $H$, let $f_H$ denote $(f_e)_{e\in E(H)}$.
We again leverage the concept of a good pair of subgraphs, which becomes much simpler to
state in the single-commodity setting.

\begin{definition}[specialization of Definition~\ref{goodgraphs}]
Given $s$-$t$ flows $f$, $\tf$, we call a pair of subgraphs $H_1$, $H_2$ 
{\em $(f,\tf)$-good} if: 
\begin{enumerate}[(i)]
\item the parallel-join of $H_1$, $H_2$ is a subgraph; 
\item $f_e\geq\tf_e$ for all $e\in E(H_1)$ and $f_e\leq\tf_e$ for all $e\in E(H_2)$; and 
\item $\val{f_{H_1}}>\val{\tf_{H_1}}$ and $\val{f_{H_2}}<\val{\tf_{H_2}}$.
\end{enumerate}
\label{defn:goodsubgraphs}
\end{definition}

\begin{lemma}
Let $g$ be any Stackelberg routing. If $f(g) + g \neq f^*$, there exists an
$(f(g)+g,f^*)$-good pair of subgraphs. 
\label{lem:gooduseful}
\end{lemma}

Lemma~\ref{lem:gooduseful} follows from a more general result proved in
Lemma~\ref{helper:sepagood} for multicommodity flows. The proof in the single-commodity
setting becomes much simpler, and follows immediately from Claim~\ref{clm:gooduseful} since $\val{f(g) + g} = \val{f^*}$.

\begin{clm}
For any two $s$-$t$ flows $f$, $\tf$ in a sepa graph $G$, either there is an $(f,\tf)$-good pair of subgraphs,
or one of the following holds:
\begin{enumerate}[(i)]
\item If $\val{f}=\val{\tf}$ then $f=\tf$.
\item If $\val{f}>\val{\tf}$ then $f\geq\tf$.
\item If $\val{f}<\val{\tf}$ then $f\leq\tf$.
\end{enumerate}
\label{clm:gooduseful}
\end{clm}

\begin{proof}
The proof is by induction on the size of the graph. For a single
edge, there is no good pair of subgraphs, but one of the three cases must hold. 
For the induction step, let $G$ be the join of subgraphs $G_1$ and $G_2$. 
Let $f_1=f_{G_1}$, $\tf_1=\tf_{G_1}$, and $f_2=f_{G_2}$, $\tf_2=\tf_{G_2}$. Clearly, 
$f_1,\tf_1$ are $s_{G_1}$-$t_{G_1}$ flows, and $f_2,\tf_2$ are $s_{G_2}$-$t_{G_2}$ flows.
If $G_1$ contains an $(f_{1},\tf_{1})$-good pair, or $G_2$ contains an $(f_2,\tf_2)$-good
pair, then the same pair is an $(f,\tf)$-good pair, and we are done. So assume otherwise. 

Suppose $G_1$ and $G_2$ are in series. 
Then, $\val{f_1}=\val{f}=\val{f_2}$, and $\val{\tf_1}=\val{\tf}=\val{\tf_2}$. So whichever  
case applies to $f$ and $\tf$, the same applies to $f_1, \tf_1$, and $f_2, \tf_2$. By the
induction hypothesis, we have the desired relationship between $f_1,\tf_1$ and
$f_2,\tf_2$, and hence between $f$ and $\tf$. So the statement holds for $G$. 
 
Suppose $G_1$ and $G_2$ are in parallel. 
If $\val{f_1}>\val{\tf_1}$ and $\val{f_2}<\val{\tf_2}$, then by the induction hypothesis, 
$f_1\geq\tf_1$, $f_2\leq\tf_2$, so $G_1$, $G_2$ is an $(f,\tf)$-good pair. Similarly, if
$\val{f_1}<\val{\tf_1}$ and $\val{f_2}>\val{\tf_2}$, then $G_2$, $G_1$ is an
$(f,\tf)$-good pair. So assume neither case holds.

Now if $\val{f}=\val{\tf}$, then (after eliminating the above cases)
$\val{f_1}=\val{\tf_1}$, $\val{f_2}=\val{\tf_2}$. Hence, by the induction hypothesis, we
have $f_1=\tf_1,\ f_2=\tf_2$, and so $f=\tf$.

If $\val{f}>\val{\tf}$, then it must be that $\val{f_1}\geq\val{\tf_1}$ and
$\val{f_2}\geq\val{\tf_2}$. Therefore, $f_1\geq\tf_1$, $f_2\geq\tf_2$, and so $f\geq\tf$.   

Finally, if $\val{f}<\val{\tf}$, then it must be that $\val{f_1}\leq\val{\tf_1}$,
$\val{f_2}\leq\val{\tf_2}$. Hence, $f_1\leq\tf_1$, $f_2\leq\tf_2$, and so $f\leq\tf$.
This completes the induction step, and hence, the proof. 
\end{proof}

Our algorithm is now quite simple to describe. 
We keep track of the set $\bar{S}$, initialized to $\es$, of 
edges not on any shortest $s$-$t$ path under the edge costs $(l^*_e(f^*_e))_e$.
By Lemma~\ref{stack-basic}, $\bar{S}$ must be saturated by any Stackelberg routing that
induces $f^*$. 
We repeatedly do the following.

{\small \vspace{8pt} \hrule 
\begin{list}{S\arabic{enumi}.}{\usecounter{enumi} \topsep=0ex \itemsep=0ex
    \addtolength{\leftmargin}{-0.5ex}} 
\item Find the flow $g$ of minimum value that saturates every edge in $\bar{S}$ and
satisfies $g_e \le f_e^*$ for all $e$. 
\item Query the oracle with $g$ as the Stackelberg flow. 
If $f^* = f(g) + g$, exit and return $g$. 
Otherwise, 
find an $(f(g)+g,f^*)$-good pair of subgraphs $H_1$, $H_2$. Add every edge in $H_2$ to $\bar{S}$ (and
repeat the process). 
\end{list}
\hrule}


\begin{theorem}
The above algorithm computes a Stackelberg flow $g$ of minimum value that induces $f^*$
in at most $m$ queries. 
\end{theorem}

\begin{proof}
In every iteration, $|\bar{S}|$ increases by at least 1: since 
$\val{f_{H_2}(g)+g_{H_2}}<\val{f^*_{H_2}}$ and $g$ saturates every edge in
$\bar{S}$, we know that at least one edge in $H_2$ is not in the current set $\bar{S}$. 
When $\bar{S} = E$, we have $g = f^*$. So the algorithm terminates in at most $m$
iterations with some flow $g$ that induces $f^*$. To complete the proof, we only need to show
that any edge added to $\bar{S}$ is indeed a non-shortest-path edge. 
Let $h=f(g)+g$. Let $s'=s_{H_1}=s_{H_2}$, $t'=t_{H_1}=t_{H_2}$. 
Since $\val{h_{H_1}}>\val{f^*_{H_1}}$, there is some $s'$-$t'$ path $P$ in $H_1$ such that
$(h-f^*)_P>0$. So $P$ belongs to a shortest $s$-$t$ path under edge costs
$(l^*_e(h_e))_e$. So for every $s'$-$t'$ path $Q$ in $H_2$, we have 
$l^*_P(f^*)<l^*_P(h)\leq l^*_Q(h)\leq l^*_Q(f^*)$. So every edge of $H_2$ is a
non-shortest-path edge under edge costs $(l^*_e(f^*_e))_e$.
\end{proof}

\section{Query- and computational- complexity lower bounds} \label{lbounds}

\subsection{A linear lower bound for query complexity with tolls} \label{lb-tolls}
We show a lower bound of $\Omega(m)$ on the number of queries required to
obtain tolls that give the target flow.  

\begin{theorem}
Any deterministic algorithm that computes tolls required to enforce a target flow requires
$\Omega(m)$ queries, even for a single commodity instance on parallel links with linear
delay functions. 
\label{thm:tolllb}
\end{theorem}

Our example for the lower bound consists of a single commodity on $m$ parallel links, with
the demand $d = m$ and the target flow $f_e^* = 1$ on each edge. In fact, our lower bound is actually for the problem of obtaining tolls $\tau$ with the right support: $f_e(\tau) > 0$ iff $f_e^* > 0$.

For our lower bound example, our delay functions are defined
by a permutation $\pi^*:[m] \rightarrow [m]$. The delay function on the parallel edge $e_i$ is be given
by $(x/m) + 2(\pi^*(i)-1)$. Thus, we use the notation $f(\pi, \tau)$ for the equilibrium flow, where the permutation $\pi$ identifies the delay functions. We show that any algorithm that computes the correct tolls to
enforce $f^*$ must obtain the correct permutation $\pi^*$, and we design an oracle that
after $k$ queries has only revealed information about ${\pi^*}^{-1}(1), \cdots,
{\pi^*}^{-1}(k)$. Thus, in order to compute the correct tolls, any algorithm requires $m-1$
queries. 

Our oracle works as follows. Initially, let $A^0 = \emptyset$ be the set of assigned edges
in the partial permutation $\pi^*$. For the $j$th query $\tau^j = (\tau_e^j)_{e \in E}$, our oracle
returns the equilibrium flow described below. 

\paragraph{Oracle:} Pick an arbitrary edge $e$ that with minimum toll $\tau_e^j$, so that $e$ is not in $A^{j-1}$. Let $\pi^*(e) = j$ and $A^j = A^{j-1} \cup
\{e\}$. Let $\pi_{(j)}$ be a complete permutation that extends the partial permutation $\pi^*$, and return
$f(\pi_{(j)},\tau^j)$ as the equilibrium flow in response to tolls $\tau^j$. 

\begin{clm}
For $j \in [m]$, let $\sigma$ be any permutation that satisfies $\sigma(e) = \pi_{(j)}(e)$
for all edges $e \in A^j$. Then for any edge $e \not \in A^j$, $f_e(\sigma,\tau^j) = 0$ 
\label{clm:tolllb}
\end{clm}

\begin{proof}
By description of the oracle and since $e \not \in A^j$, $e$ is not the unique edge with
minimum toll in $\tau_j$. That is, there exists an edge $h \in A^j$ with $\tau_h^j \le
\tau_e^j$. Since $h \in A^j$, $\sigma(h) \le j < \sigma(e)$. Further since $\tau_h^j \le
\tau_e^j$, by description of the delay functions, $l_e(0) > l_h(d)$, and hence edge $e$
cannot have flow at equilibrium. 
\end{proof}

We now show that the equilibrium flows returned by our oracle are consistent.

\begin{lemma}
There exists a permutation $\sigma$ so that for every $j \in [m]$, $f(\pi_{(j)},\tau^j) = f(\sigma,\tau^j)$.
\label{lem:tolllb}
\end{lemma}

\begin{proof}
Fix $j \le m$, and let $\sigma$ be a complete permutation that extends $\pi^*$. The image of every edge $e \in A^j$ is the same in $\pi^*$ and $\pi_{(j)}$. Thus the delay functions on both edges is the same. By Claim~\ref{clm:tolllb}, the equilibrium is zero for any edge not in $A^j$. Since edges in $A^j$ have the same delay function, the equilibrium flow must be the same for permutations $\sigma$ and $\pi_{(j)}$.
\end{proof}

\noindent \emph{Proof of Theorem~\ref{thm:tolllb}.} From Lemma~\ref{lem:tolllb}, for any
sequence of $m-1$ toll queries, the oracle returns a consistent sequence of equilibrium
responses. Further, from Claim~\ref{clm:tolllb}, since $|A^j| \le m-1$ for $j \le m-1$,
there is an edge $e$ with no flow in the equilibrium returned by the oracle. Hence, since
$f^*$ has positive flow on every edge, any deterministic algorithm requires at least $m$
queries to compute tolls that obtain $f^*$. \qed

\subsection{Lower bounds for determining equivalence with Stackelberg routing}
\label{sec:stackelberglb} \label{lb-stackelberg}

Given the ability to query a routing game and obtain the equilibrium flow, a natural
question is if we can in fact obtain the delay functions on the edges. It is obvious that
the exact delay functions cannot be obtained, even for a single edge. However, is it
possible to obtain delay functions that are equivalent, in the sense that any 
Stackelberg routing would yield almost the same equilibrium flow as in the routing game? 

\begin{definition}
Given a graph $G$ with demand $d$ between $s$ and $t$ and a Stackelberg demand fraction $\alpha$, two sets of delay functions on the edges $l^1$ and $l^2$ are $\epsilon$-equivalent if for
every Stackelberg routing $g$ with $|g| \le \alpha d$, $\|f(l^1,g) - f(l^2,g)\|_\infty \le
\epsilon$. 
\label{defn:equivalentdelays}
\end{definition}

We prove strong lower bounds for this problem, both for the query complexity and the computational complexity. In fact, for the query complexity, the lower bound instance is a graph of constant size. The size of the input is determined by the demand $d$, and we show
that although the size of the input is $O(\log d)$ any deterministic algorithm that
determines $\epsilon$-equivalence for a fixed $\epsilon$ must make $\Omega(\sqrt{d})$
queries. For computational hardness, we show that even if we are
explicitly given affine delay functions $l^1$ and $l^2$, determining $1/2$-equivalence is \nphard.  Our proof for computational hardness builds upon a reduction given by Roughgarden~\cite{Roughgarden06}.


\paragraph{Query complexity.} We are now given a graph $G$ with demand $d$ between $s$ and $t$, a Stackelberg demand fraction $\alpha$, and a set of delay functions $l^1$ on the edges of $G$. In addition, we are given query access to a second set of delay functions $l^2$. As before, each query consists of a Stackelberg routing $g$, and the response is the equilibrium flow $f(l^2,g)$. We show the following result.

\begin{figure}[h]
\centering
\psfrag{e1}[bc]{$x^2+x$}
\psfrag{e2}[tr]{$ax$}
\psfrag{e3}{$b$}
\psfrag{e4}[bl]{$ax$}
\psfrag{e5}[tl]{$x^2+x$}
\psfrag{s}{$s$}
\psfrag{t}{$t$}
\psfrag{u}{$u$}
\psfrag{v}{$v$}
\psfrag{d}{$d$}
\includegraphics[scale=0.25]{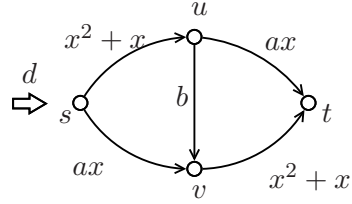}
\caption{Braess graph instance for proving hardness of equivalence determination.}
\label{fig:braessstackelberg2}
\end{figure}

\begin{theorem}
Any deterministic algorithm that determines $\epsilon$-equivalence for $\epsilon \le 1/16$ requires an exponential number of queries. 
\label{thm:stackelbergequivalence}
\end{theorem}

Our proof of the theorem is
based on a particular property exhibited by the Braess graph shown in
Figure~\ref{fig:braessstackelberg2}: there exist demands $d_1 \le d_2$ that depend on the parameters $a$ and $b$ so that for any demand $d < d_1$ and $d > d_2$ the
set of shortest-path edges is the same, and differs from the set of shortest-path edges for any demand $d_1 \le d \le d_2$. This is formalized by the following claim.

\begin{clm}
For the routing game depicted in Figure~\ref{fig:braessstackelberg2}, and any $d_1$, $d_2 \in \mathbb{R}_+$ with $d_2 > d_1 \ge 1$, there exist parameters $a$ and $b$ so that the equilibrium flow $f$ on
the Braess edge is strictly positive iff $d_1 < d < d_2$, where $d$ is the
demand being routed. Further, if $d_2 - d_1 \ge \sqrt{2(d_1 + d_2)}$, then $f_{uv} \ge
1/12$ for demand $d = (d_1 + d_2)/2$. 
\label{clm:braessclaim}
\end{clm}

\begin{proof}
We choose $a =1+(d_1 + d_2)/2$ and $b = (d_1 d_2)/4$. Then 
for any $d$, consider the flow that routes $d/2$ on the $s$-$u$-$t$ path and $d/2$ on the
$s$-$v$-$t$ path. It is easy to verify that this is the equilibrium flow if and only if $d
\le d_1$ or $d \ge d_2$. Given the symmetric delay functions, it is then
apparent that for $d \in (d_1, d_2)$ the $(u,v)$ edge must have strictly
positive flow. 

For the second part of the proof, let $\sigma = d_1 + d_2$, $\delta = d_2 -
d_1$. Thus $a = 1 + \sigma/2$, $b = (\sigma^2 - \delta^2)/16$, and $d = \sigma/2 =
a-1$. Let $f_{su} = x$. Then by the symmetry of the delay functions $f_{sv} = f_{ut} =
d-x$ and $f_{uv} = 2x-d$. Since we know for this demand $f_{uv} > 0$, and edge $(s,v)$ has
zero delay if $f_{sv} = 0$, $l_{su}(f) + l_{uv}(f) - l_{sv}(f) = 0$. Hence 

\[
0 ~ = x^2 + x + b - a(d-x) 
\]

\noindent and solving for $f_{uv} = 2x-d$, and substituting the values of $a$, $b$ and $d$ yields

\begin{align*}
2x-d & = \sqrt{(a+1)^2 +4(ad-b)} - (a+1+d) ~ = x^2 + (a+1)x + (b-ad) \\
	& = \sqrt{(2+\sigma)^2 + \delta^2/4} - (2+\sigma) \, .
\end{align*}

\noindent Using the fact that $\sqrt{1+x} \ge 1 + x/3$ for $|x| \le 1$ by the Taylor expansion, we get

\begin{align*}
2x-d \ge \frac{\delta^2}{12(2+\sigma)}
\end{align*}

\noindent Since $\delta \ge \sqrt{2\sigma} \ge \sqrt{2+\sigma}$ by assumption, this completes the proof.
\end{proof}

\begin{proofof}{Theorem~\ref{thm:stackelbergequivalence}}
We demonstrate that on the Braess graph in Figure~\ref{fig:braessstackelberg2} with an
additional $(s,t)$ edge, demand $d > 8$, and where $a$, $b$ have value $O(d^2)$, any
algorithm requires $\Omega(\sqrt{d})$ queries to determine if two sets of delay functions
$l^1$, $l^2$ are equivalent. Since the size of the input is $O(\log d)$, this would prove
the lemma. 

For delay functions $l^1$ that are explicitly given, $a=1$, $b=
\infty$ and $l^1_{st} = \infty$. Delay functions $l^2$ also have $l^2_{st} = \infty$ but
different values for $a$ and $b$, which are determined after seeing the queries. Let
$g^i$, $i \le k$ be the set of queries. We will show that if $k \le \sqrt{d}$ then there
exist $a$, $b$ so that $f(l^1,g^i) = f(l^2,g^i)$ for all $i \le k$, but there exists $g$
so that $f_{uv}(l^2,g^i) \ge 1/12$. Since $l_{uv}^1 = \infty$, it must be that $f_{uv}(l^1,g^i) =
0$, and hence the two delay functions are distinct. Thus any algorithm that makes less
than $\sqrt{d}$ must fail to distinguish between these delay functions. 

For any query $g^i$, our oracle returns the equilibrium flow $f(l^1, g^i)$. Now given
$g^i$ for $i \le k \le \sqrt{d}$, let $\alpha_1$, $\alpha_2 \in [1,d]$ be such that
$\alpha_2 - \alpha_1 \ge \sqrt{d}$, and for all $i$, $d-g_{st}^i \not \in [\alpha_1,
  \alpha_2]$. Since edge $(s,t)$ has infinite delay, any flow on this edge must be
Stackelberg flow. Hence we require $\alpha_1$ and $\alpha_2$ so that the total flow on the
Braess graph is always outside the interval $[\alpha_1, \alpha_2]$, and $\alpha_2 -
\alpha_1 \ge \sqrt{d}$. Since $k \le \sqrt{d}$, such an interval must exist. We then
select $a$ and $b$ as in Claim~\ref{clm:braessclaim} to complete our definition of delay
function $l^2$. 

It remains to show that for all $g^i$, $f(l^2, g^i) = f(l^1, g^i)$ for correctness of the
oracle. Fix $i$, and let $d' = d-g_{st}^i$. Note that $d' \not \in [\alpha_1,
  \alpha_2]$. Let $g_1$, $g_2$ and $g_3$ be the Stackelberg flow on paths $s$-$u$-$t$,
$s$-$u$-$v$-$t$, and $s$-$v$-$t$ respectively. By our choice of $l^2$, if $g_1=g_2=g_3 =
0$, then the equilibrium flow would split demand $d'$ equally between the $s$-$u$-$t$ and
$s$-$v$-$t$ paths, and hence  

\begin{align}
l^2_{sv}(d'/2) & \le b + l^2_{su}(d'/2) \, .
\label{eqn:flowsplit}
\end{align}

\noindent We consider the following cases.

\noindent \textbf{Case 1: Either $g_1$ or $g_3$ is strictly greater than
  $(d'-g_2)/2$}. Suppose $g_1 > (d'-g_2/2)$. We claim that at equilibrium, the non-Stackelberg demand is entirely routed on the $s$-$v$-$t$ path, i.e., $f(g) = h$ where $h_{sv}(g) = h_{vt}(g) = d'-(g_1+g_2+g_3)$. To see this, note that $h_{sv} + g_{sv} < d'/2$, hence comparing with~\eqref{eqn:flowsplit}, delay on the $s$-$v$-$t$ path is less than the delay on the $s$-$u$-$v$-$t$. Further, $h_{sv} + g_{sv} < g_{ut}$, and $h_{vt} + g_{vt} < g_{su}$. By the symmetry of delay functions, the $s$-$v$-$t$ path is therefore the shortest path, and hence $h = f(g)$.

\noindent \textbf{Case 2: Both $g_1$ and $g_3$ are at most $(d' - g_2)/2$}. We claim that
at equilibrium, $f_{su}(g) = f_{ut}(g)$ $= d'/2 - (g_1 + g_2)$ and $f_{sv}(g) = f_{vt}(g)$
$= d'/2 - (g_3 + g_2)$. Thus at equilibrium the remaining flow $d'-g_2$ is divided equally
between the $s$-$u$-$t$ and $s$-$v$-$t$ paths, and again the edge $(u,v)$ has no flow at
equilibrium. To verify the claim, note that  

\[
f_{sv}(g) + g_{sv} \le f_{su}(g) + g_{su} = d'/2 \mbox{ and } f_{ut}(g) + g_{ut} \le f_{vt}(g) + g_{vt} = d'/2
\]

\noindent and hence, comparing with~(\ref{eqn:flowsplit}),

\[
l^2_{sv}(f(g) + g) \le l^2_{sv}(f(g) + g) + b \mbox{ and } l^2_{ut}(f(g) + g) \le l^2_{vt}(f(g) + g) + b \, .
\]

\noindent It is further easy to see that, since the total flow on edges $(s,u)$, $(v,t)$
is equal, and the total flow on edges $(s,v)$, $(u,t)$ is equal, 

\[
l^2_{sv}(f(g) + g) + l^2_{vt}(f(g) + g) = l^2_{su}(f(g) + g) + l^2_{ut}(f(g) + g) \, .
\]

\noindent Hence paths $s$-$u$-$t$ and $s$-$v$-$t$ are shortest paths with the described
flow, and since $f_P(g) > 0$ only on these paths, it is an equilibrium. 

As noted earlier, if the Stackelberg flow is rational, then so is $f(g)$. In fact as shown the equilibrium flow in all cases is very simple and can be computed directly from $g$.
\end{proofof}

We note that in our example, the equilibrium flow returned by the oracle is particularly simple and in fact does not depend on the delay functions. E.g., in the simpler case in the proof sketch, the oracle always returns $f_e(g) = (d-g_{st})/2$ for all $e \neq (s,t)$, $(u,v)$.

\paragraph{Computational complexity.} We now show that even if delay functions $l^1$ and $l^2$ are given explicitly, determining if they are $\epsilon$-equivalent is
computationally hard for $\epsilon \le 1/2$. This is true even if all delay functions are
affine. Our proof uses properties of the Braess graph together with ideas from a reduction from 2-Directed Disjoint Paths shown by Roughgarden~\cite{Roughgarden06}.

\begin{definition}[2-Directed Disjoint Paths (2DDP)]
Given a directed graph $G=(V,E)$ and two pairs of terminals $s_1$,$t_1$ and $s_2$, $t_2$,
determine if there exist $s_i$-$t_i$ paths $p_i$ so that $p_1$ and $p_2$ are
vertex-disjoint. 
\end{definition}

\begin{figure}[h]
\centering
\psfrag{e1}[br]{$x$}
\psfrag{e2}[tr]{$1$}
\psfrag{e3}{$0$}
\psfrag{e4}[bl]{$1$}
\psfrag{e5}[tl]{$x$}
\psfrag{s}{$s$}
\psfrag{t}{$t$}
\psfrag{u}{$u$}
\psfrag{v}{$v$}
\psfrag{d}{$d$}
\includegraphics[scale=0.25]{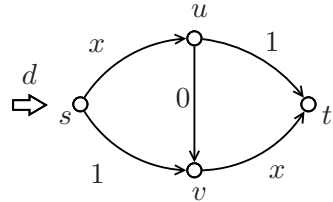}
\caption{Braess graph instance for proving hardness of equivalence determination with respect to Stackelberg routing.}
\label{fig:braessstackelberg}
\end{figure}

\begin{figure}[h]
\centering
\psfrag{s}{$s$}
\psfrag{sp}{$s'$}
\psfrag{s1}{$s_1$}
\psfrag{s2}{$s_2$}
\psfrag{t1}{$t_1$}
\psfrag{t2}{$t_2$}
\psfrag{t}{$t$}
\psfrag{x}{$x$}
\psfrag{y}{$y$}
\psfrag{18}{$1/8$}
\psfrag{infty}{$\infty$}
\psfrag{x}{$x$}
\psfrag{1}{$1$}
\psfrag{xy}{$x/m^2$}
\includegraphics[scale=0.5]{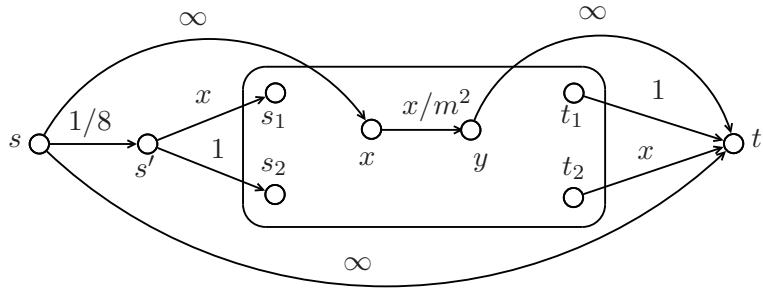}
\caption{2DDP instance with additional edges for proving hardness of equivalence determination with respect to Stackelberg routing.}
\label{fig:ddp2}
\end{figure}

\begin{theorem}
The problem of determining the $\epsilon$-equivalence of delay functions for $\epsilon \le 1/2$ is \nphard. 
\label{thm:stackelbergeq}
\end{theorem}

We use the following claim about Stackelberg routing in the Braess graph.

\begin{clm}
In any Stackelberg routing instance on the graph with delay functions $l$ as in
Figure~\ref{fig:braessstackelberg} and Stackelberg routing $g$, if $f_{uv}(g) > 0$, then
$d < 2$ and $D(l,g) \le 2$.  Further, if $g_e = 0$ for every edge in the Braess graph and the demand $d \in [1/2, 3/2]$, then $f_{uv}(g) \ge 1/2$.
\label{clm:braessedge}
\end{clm}

\begin{proof}
For the first part of the claim, het $h = f(l,g)+g$. If $(u,v) \in \mcs(l,g)$, then the path $p=(s,u,v,t)$ must be a
shortest path for flow $h$. Then $l_{su}(h) + l_{uv}(h) \le l_{sv}(h)$, and hence
$l_{su}(h) \le 1$. Thus $h_{su} \le 1$. Similarly, $h_{vt} \le 1$. The first part follows. For the second part, if $d \le 1$ it is easy to see that the equilibrium flow routes the entire demand on the $s$-$u$-$v$-$t$ path. If $d \in [1,2]$ then consider the flow $h_{su} = h_{vt} = 1$, $h_{uv} = 2-d$ and $h_{sv} = h_{ut} = d-1$. It can be verified that $h$ is the equilibrium flow.
\end{proof}

\begin{proofof}{Theorem~\ref{thm:stackelbergeq}} We show a reduction from 2DDP. Given an instance of the 2DDP problem, after the addition of a
source $s$ and a sink $t$ and additional edges described next (and shown in
Figure~\ref{fig:ddp2}), we add this graph in parallel with a standard Braess graph
(Figure~\ref{fig:braessstackelberg}). The delay functions $l^1$, $l^2$ will differ only on
edge $(u,v)$ in the Braess graph. We use $H_1$ to refer to the Braess graph and $H_2$ to
refer to the graph in the 2DDP instance with vertices $s$ and $t$ and the additional
edges, and $H$ to refer to their parallel composition. For a flow $f$, $|f_{H_i}|$ is the
value of the flow in subgraph $H_i$.

The specifics of the construction are as follows. Let $m=|E|$ be the number of edges in
the given instance of 2DDP. All of these edges have delay function $x/m^2$. We add a
source $s$, vertex $s'$ and a sink $t$. We add an edge $(s,s')$ with constant delay
function $1/8$, and edges $(s',s_i)$ and $(t_i,t)$ for $i \in \{1,2\}$. Edges $(s,s_1)$
and $(t_2,t)$ have delay function $x$, while edges $(s,s_2)$ and $(t_1,t)$ have delay
function 1. Further, there is an $(s,t)$ edge with delay function $\infty$, and for every
edge $e=(x,y)$ in the original instance, the new instance additionally contains edges
$(s,x)$ and $(y,t)$ with delay function $\infty$. This constitutes the graph $H_2$. Graph
$H_1$ consists of the Braess graph instance in Figure~\ref{fig:braessstackelberg}, and
graph $H$ is obtained by a parallel composition of $H_1$ and $H_2$. The Stackelberg
instance has demand $m^4 + 3$, and $\alpha = m^4/(3+m^4)$. The delay functions $l^1$,
$l^2$ are as described, except $l^2_{uv} = \infty$ on the Braess edge. 

Since $l^1$, $l^2$ differ only
on the delay function on edge $(u,v)$, it is easy to see that for any Stackelberg routing
$g$, if $f_{uv}(l^1,g) = 0$ then $f_{uv}(l^1,g) = f_{uv}(l^2,g)$. Further, since $l^2_{uv}
= \infty$, if $f_{uv}(l^1,g) = f_{uv}(l^2,g)$ then in fact $f_{uv}(l^1,g) = 0$. Hence the
delay functions are equivalent iff $f_{uv}(l^1,g) = 0$ for every Stackelberg routing
$g$. For the proof of the theorem, we will show that if the instance of 2DDP is a positive
instance, then there exists a Stackelberg routing $g$ so that $f_{uv}(l^1,g) \ge 1/2$,
otherwise for any Stackelberg routing, $f_{uv}(l^1,g) = 0$. 

In the remainder of the proof we focus on delay functions $l^1$. Suppose that the instance
is a positive instance. Then the Stackelberg routing $g$ sends $m^3$ units of flow on
every edge $e=(x,y)$ in the original instance that is not on the vertex-disjoint paths
$p_i$, using the additional edges $(s,v)$, $(v,w)$. Any remaining Stackelberg flow is
routed on the $(s,t)$ edge. Thus every edge that is not on the vertex-disjoint paths now
has delay at least $m$, while $g_e = 0$ for  every edge on the vertex-disjoint
paths. Further, $g_e = 0$ for every edge $e \in H_1$. 

We claim that for the equilibrium flow, $1/2 \le |f_{H_1}| \le 3/2$. To see this, if
$|f_{H_1}| < 1/2$, then the delay at equilibrium in $H_1$ is at most $1$. However
$|f_{H_2}| > 5/2$, hence at least one of the two $s$-$t$ parallel paths has delay at
equilibrium greater than 1. If $|f_{H_1}| > 3/2$ then the delay at equilibrium in $H_1$ is
2. However, $|f_{H_2}| \le 3/2$, hence at least one of the two $s$-$t$ parallel paths has
delay at equilibrium less than $1 + 1/8 + 3/4 \times (1+1/m)$ $< 2$. Thus at equilibrium,
$1/2 \le |f_{H_1}| \le 3/2$, and by Claim~\ref{clm:braessedge}, $f_{uv}(g) \ge 1/2$. 

Now suppose that for some Stackelberg routing $g$, $f_{uv}(g) > 0$. Then by
Claim~\ref{clm:braessedge}, $|f_{H_1}| \le 2$ and the delay at equilibrium is at most
2. However, then $|f_{H_2}| \ge 1$ and the delay at equilibrium is at most 2. Since there
is an $(s,s')$ edge with constant delay $1/8$, following the proof of
Theorem~\ref{thm:eqinapprox}, this is only possible if the instance of DDP is a positive
instance.
\end{proofof}

In fact, using very similar ideas, we can show that the problem of minimizing $D(f(g))$ over all Stackelberg strategies is $(4/3-\epsilon)$-inapproximable, even with linear delays. Roughgarden has shown that finding the Stackelberg routing that minimizes the
average delay of the \emph{total flow} $g+f(g)$ is \nphard, even in parallel links with
affine delays~\cite{Roughgarden04}. Despite considerable interest in Stackelberg routing,
nothing stronger than NP-hardness is known for this problem. Our result thus shows that a closely related problem is APX-hard.

\begin{definition}[Stackelberg Equilibrium Delay Minimization (SEDM)]
Given a Stackelberg routing instance $(G,l,(d,s,t),\alpha)$, find the Stackelberg routing
$g$ that minimizes the average delay for the equilibrium flow $f(g)$. 
\end{definition}

\begin{theorem}
The SEDP problem is $(4/3-\epsilon)$-inapproximable, for any fixed $\epsilon > 0$. 
\label{thm:eqinapprox}
\end{theorem}

\begin{proof}
Given an instance of the 2DDP problem, we modify it to obtain a Stackelberg routing
instance as follows. Let $m=|E|$ be the number of edges in the original instance. All of
these edges have delay function $x/m^2$. We add a source $s$ and a sink $t$, and edges
$(s,s_i)$ and $(t_i,t)$. Edges $(s,s_1)$ and $(t_2,t)$ have delay function $x$, while
edges $(s,s_2)$ and $(t_1,t)$ have delay function 1. Further, there is an $(s,t)$ edge
with delay function $\infty$, and for every edge $e=(v,w)$ in the original instance, the
new instance additionally contains edges $(s,v)$ and $(w,t)$ with delay function
$\infty$. The Stackelberg instance has demand $m^4 + 1$, and $\alpha = m^4/(1+m^4)$. 

We claim that if the instance of 2DDP is a positive instance, then there exists a
Stackelberg routing $g$ with $D(l,f(g)) \le 3/2 + 1/m$, otherwise for any Stackelberg
routing, $D(f(g)) \ge 2$. Suppose that the instance is a positive instance. Then the
Stackelberg routing $g$ sends $m^3$ units of flow on every edge $e=(v,w)$ in the original
instance that is not on the vertex-disjoint paths $p_i$, using the additional edges
$(s,v)$, $(v,w)$. Thus every edge that is not on the vertex-disjoint paths now has delay
at least $m$, while $g_e = 0$ for  every edge on the vertex-disjoint paths. Any remaining
Stackelberg flow is routed on the $(s,t)$ edge. It is now easy to verify that the
equilibrium flow $f(g)$ splits one unit of demand approximately equally between the two
paths $p_1$ and $p_2$, and has a delay at equilibrium of at most $3/2 + 1/m$. 

Now suppose the given instance does not contain two vertex-disjoint paths between $s_1$,
$t_1$ and $s_2$, $t_2$. Following the argument in~\cite{BhaskarLS13}, for a contradiction let $g$ be a Stackelberg routing for which $D(f(g)) < 2$. Let $F$ be the set
of edges with positive flow at equilibrium. Then $F$ must contain all four edges
$(s,s_1)$, $(s,s_2)$, $(t_1,t)$, $(t_2,t)$; the absence of any of these edges would give
a delay of at least 2. Further, $F$ cannot contain an $s$-$s_2$-$t_1$-$t$ path
since again this would given delay of at least 2. Hence $F$ must contain an
$s$-$s_1$-$t_1$-$t$ path and an $s$-$s_2$-$t_2$-$t$ path. These paths cannot be vertex
disjoint; let $v$ be the common vertex. Then the delay on any $s$-$v$ path must be
at least 1, and the delay on any $v$-$t$ path must be at least 1. Hence the total
delay in any instance that does not contain two vertex-disjoint paths is at least 2, which gives us a contradiction. The
hardness of determining the existence of these paths thus shows that minimizing the delay
of the equilibrium flow with Stackelberg routing is $(4/3-\epsilon)$ inapproximable, for any
$\epsilon > 0$. 
\end{proof}

\paragraph{\large Acknowledgment.} 
We thank \'{E}va Tardos for useful discussions.

\bibliographystyle{plain}
\bibliography{ir-biblio}

\begin{thebibliography}{10}

\bibitem{BabichenkoBP14}
Yakov Babichenko, Siddharth Barman, and Ron Peretz.
\newblock Simple approximate equilibria in large games.
\newblock In {\em EC}, pages 753--770, 2014.

\bibitem{BeckmannMW56}
Martin Beckman, CB~McGuire, and Christopher~B Winsten.
\newblock Studies in the economics of transportation.
\newblock Technical report, 1956.

\bibitem{BhaskarFHH09}
Umang Bhaskar, Lisa Fleischer, Darrell Hoy, and Chien-Chung Huang.
\newblock Equilibria of atomic flow games are not unique.
\newblock In {\em Proceedings of the twentieth Annual ACM-SIAM Symposium on
  Discrete Algorithms}, pages 748--757. Society for Industrial and Applied
  Mathematics, 2009.

\bibitem{BhaskarLS13}
Umang Bhaskar, Katrina Ligett, and Leonard~J Schulman.
\newblock The network improvement problem for equilibrium routing.
\newblock {\em arXiv preprint arXiv:1307.3794}, 2013.

\bibitem{BonifaciHS10}
Vincenzo Bonifaci, Tobias Harks, and Guido Sch{\"a}fer.
\newblock {S}tackelberg routing in arbitrary networks.
\newblock {\em Math. Oper. Res.}, 35(2):330--346, 2010.

\bibitem{ColeDR03}
R.~Cole, Y.~Dodis, and T.~Roughgarden.
\newblock Pricing network edges for heterogeneous selfish users.
\newblock In {\em Proceedings of the 35th Annual ACM Symposium on Theory of
  Computing}, pages 521--530, 2003.

\bibitem{ColeDR06}
Richard Cole, Yevgeniy Dodis, and Tim Roughgarden.
\newblock How much can taxes help selfish routing?
\newblock {\em J. Comput. Syst. Sci.}, 72(3):444--467, 2006.

\bibitem{DafermosN84}
Stella Dafermos and Anna Nagurney.
\newblock Sensitivity analysis for the asymmetric network equilibrium problem.
\newblock {\em Mathematical programming}, 28(2):174--184, 1984.

\bibitem{FearnleyGGS13}
John Fearnley, Martin Gairing, Paul~W. Goldberg, and Rahul Savani.
\newblock Learning equilibria of games via payoff queries.
\newblock In {\em ACM Conference on Electronic Commerce}, pages 397--414, 2013.

\bibitem{Fleischer05}
Lisa Fleischer.
\newblock Linear tolls suffice: New bounds and algorithms for tolls in single
  source networks.
\newblock {\em Theor. Comput. Sci.}, 348(2-3):217--225, 2005.

\bibitem{FleischerJM04}
Lisa Fleischer, Kamal Jain, and Mohammad Mahdian.
\newblock Tolls for heterogeneous selfish users in multicommodity networks and
  generalized congestion games.
\newblock In {\em FOCS}, pages 277--285, 2004.

\bibitem{ellipsoidbook}
Martin Gr{\"o}tschel, L{\'a}szl{\'o} Lov{\'a}sz, and Lex Schrijver.
\newblock Geometric algorithms and combinatorial optimization.
\newblock {\em Algorithms and Combinatorics}, 2:1--362, 1993.

\bibitem{Harks11}
Tobias Harks.
\newblock {S}tackelberg strategies and collusion in network games with
  splittable flow.
\newblock {\em Theory Comput. Syst.}, 48(4):781--802, 2011.

\bibitem{HartN13}
S.~Hart and N.~Nisan.
\newblock The query complexity of correlated equilibria.
\newblock {CS} ar{X}iv, 2013.

\bibitem{JiangB13}
Albert Jiang and Kevin Leyton-Brown.
\newblock Polynomial-time computation of exact correlated equilibria in compact
  games.
\newblock {\em Games and Economic Behavior}, 2013.
\newblock To appear.

\bibitem{KaporisS09}
A.~Kaporis and P.~Spirakis.
\newblock The price of optimum in {S}tackelberg games on arbitrary single
  commodity networks and latency functions.
\newblock {\em Theoretical Computer Science}, 410:745--755, 2009.

\bibitem{KarakostasK04}
George Karakostas and Stavros~G. Kolliopoulos.
\newblock Edge pricing of multicommodity networks for heterogeneous users.
\newblock In {\em FOCS}, pages 268--276, 2004.

\bibitem{KoutsoupiasP99}
Elias Koutsoupias and C.~Papadimitriou.
\newblock Worst-case equilibria.
\newblock In {\em Proceedings of 16th STACS}, pages 404--413, 1999.

\bibitem{KumarM02}
VS~Anil Kumar and Madhav~V Marathe.
\newblock Improved results for {S}tackelberg scheduling strategies.
\newblock In {\em Automata, Languages and Programming}, pages 776--787.
  Springer, 2002.

\bibitem{Papadimitriou01}
C.~H. Papadimitriou.
\newblock Algorithms, games, and the internet, 2001.

\bibitem{Pigou20}
A.~C. Pigou.
\newblock {\em The Economics of Welfare}.
\newblock Macmillan, 1920.

\bibitem{Rosen65}
J~Ben Rosen.
\newblock Existence and uniqueness of equilibrium points for concave n-person
  games.
\newblock {\em Econometrica: Journal of the Econometric Society}, pages
  520--534, 1965.

\bibitem{Roughgarden05}
T.~Roughgarden.
\newblock {\em Selfish Routing and the Price of Anarchy}.
\newblock MIT Press, 2005.

\bibitem{Roughgarden08}
T.~Roughgarden.
\newblock Routing games.
\newblock In N.~Nisan, T.~Roughgarden, \'{E}. Tardos, and V.~Vazirani, editors,
  {\em Algorithmic Game Theory}. Cambridge University Press, 2007.

\bibitem{Roughgarden03}
Tim Roughgarden.
\newblock The price of anarchy is independent of the network topology.
\newblock {\em J. Comput. Syst. Sci.}, 67(2):341--364, 2003.

\bibitem{Roughgarden04}
Tim Roughgarden.
\newblock {S}tackelberg scheduling strategies.
\newblock {\em SIAM J. Comput.}, 33(2):332--350, 2004.

\bibitem{Roughgarden06}
Tim Roughgarden.
\newblock On the severity of {B}raess's paradox: Designing networks for selfish
  users is hard.
\newblock {\em J. Comput. Syst. Sci.}, 72(5):922--953, 2006.

\bibitem{PapadimitriouR08}
Tim Roughgarden and Christos Papadimitriou.
\newblock Computing correlated equilibria in multi-player games.
\newblock {\em J. ACM}, 55(3):14, 2008.

\bibitem{RoughgardenS11}
Tim Roughgarden and Florian Schoppmann.
\newblock Local smoothness and the price of anarchy in atomic splittable
  congestion games.
\newblock In {\em SODA}, pages 255--267, 2011.

\bibitem{RoughgardenT02}
Tim Roughgarden and {\'E}va Tardos.
\newblock How bad is selfish routing?
\newblock {\em J. ACM}, 49(2):236--259, 2002.

\bibitem{ShmoysS06}
D.~Shmoys and C.~Swamy.
\newblock An approximation scheme for stochastic linear programming and its
  application to stochastic integer programs.
\newblock {\em Journal of the ACM}, 58:25, 2011.

\bibitem{SurekaW05}
Ashish Sureka and Peter~R. Wurman.
\newblock Using {T}abu best-response search to find pure strategy {N}ash
  equilibria in normal form games.
\newblock In {\em Proceedings of the Fourth International Joint Conference on
  Autonomous Agents and Multiagent Systems}, AAMAS '05, pages 1023--1029, New
  York, NY, USA, 2005. ACM.

\bibitem{Swamy12}
Chaitanya Swamy.
\newblock The effectiveness of {S}tackelberg strategies and tolls for network
  congestion games.
\newblock {\em ACM Transactions on Algorithms}, 8(4):36, 2012.

\bibitem{Wardrop52}
John~Glen Wardrop.
\newblock Some theoretical aspects of road traffic research.
\newblock In {\em ICE Proceedings: Engineering Divisions}, volume~1, pages
  325--362. Thomas Telford, 1952.

\bibitem{Wellman06}
Michael~P. Wellman.
\newblock Methods for empirical game-theoretic analysis.
\newblock In {\em AAAI}, pages 1552--1556, 2006.

\bibitem{YangH04}
Hai Yang and Hai-Jun Huang.
\newblock The multi-class, multi-criteria traffic network equilibrium and
  systems optimum problem.
\newblock {\em Transportation Research Part B: Methodological}, 38(1):1--15,
  2004.

\bibitem{YangZ08}
Hai Yang and Xiaoning Zhang.
\newblock Existence of anonymous link tolls for system optimum on networks with
  mixed equilibrium behaviors.
\newblock {\em Transportation Research Part B: Methodological}, 42(2):99--112,
  2008.

\end{thebibliography}


\appendix

\section{Proofs from Section~\ref{prelim}} \label{append-prelim}

\begin{proofof}{Lemma~\ref{stack-basic}}
The necessity of the first condition, that $g_e \le f_e^*$ on every edge, is obvious. For
the necessity of the second condition, assume $f^*$ is an equilibrium flow and on some
edge $e \not \in S$, $g_e < f_e^*$. Then $f_e(g) > 0$ since $f(g) + g = f^*$. By
definition of Wardrop equilibrium, then there must exist a path $P$ with $e \in P$ and
$l_P(f^*) \le l_Q(f^*)$ for any path $Q$. This contradicts that $e \not \in S$. 

For the sufficiency of the conditions, consider the flow $f^* - g$. This is strictly
positive only on shortest-path edges, and hence satisfies the conditions for Wardrop
equilibrium with Stackelberg flow $g$. Since the equilibrium is unique, $f(g) = f^* - g$. 
\end{proofof}

\section{Proofs from Sections~\ref{tolls-ellipsoid} and~\ref{tolls-extn}} 
\label{append-tollsextn}

\begin{proofof}{Lemma~\ref{e-eqbmcomp}}
Let $(G,l,(s_i,t_i,d_i)_{i\leq k})$ be the given routing game.
Recall that we assume that the $l_e$s satisfy
\eqref{invkcont}--\eqref{kgrowth} with $\log K=\poly(\I)$. 
Recall the convex program \eqref{eqn:wardrop} used to compute the Wardrop equilibrium. 
\begin{equation*}
\min \ \Phi(f):=\sum_e \int_0^{f_e} l_e(x) \, dx \quad \text{s.t.} \quad 
f = \sum_{i=1}^k f^i, \quad f^i \text{ is an $s_i$-$t_i$ flow of value $d_i$} \ \
\forall i=1,\ldots,k. \tag{\ref{eqn:wardrop}}
\end{equation*}
Set $\dt=\frac{\e}{4mK}$ and 
$\ve=\min\bigl\{\frac{\e(\sum_i d_i)}{2},\frac{\dt^2}{2K^2}\bigr\}$. 
Let $\htf$ be the Wardrop equilibrium, and $g$ be a feasible flow such that
$\Phi(g)\leq\Phi(\htf)+\ve$ that we compute in time 
$\poly\bigl(\I,\log(\frac{1}{\ve})\bigr)=
\poly\bigl(\I,\log(\frac{1}{\e})\bigr)$. 
(We will later require that $g$ is computed via a specific algorithm for solving
\eqref{eqn:wardrop}.)  

First, we note that given any feasible flow $g$, one can always obtain an acyclic feasible
flow $g'\leq g$ by simply canceling flow along flow-carrying cycles (of each
commodity). So in the sequel, we ignore the acyclicity condition and concentrate on
obtaining an approximate equilibrium.

Observe that for any feasible flows $h, f$, we have $\Phi(h)-\Phi(f)\geq v_f\cdot(h-f)$, 
$v_f=(l_e(f_e))_e$; $v_f$ is called the {\em subgradient} of $\Phi$ at $f$.
So we have 
$$
\sum_e g_el_e(\htf_e)-\sum_i d_iD^i(l,\htf)=
\sum_e(g_e-\htf_e)l_e(\htf_e)\leq\Phi(g)-\Phi(\htf)\leq\ve. \label{ineq1}
$$
We show below that $\sum_e(g_e-\htf_e)l_e(g_e)\leq\frac{\dt^2}{K}$.
Since $\htf$ is an equilibrium, we also have 
$\sum_e (\htf_e-g_e)l_e(\htf_e)\leq 0$. Adding the two inequalities gives
$\sum_e(g_e-\htf_e)\bigl(l_e(g_e)-l_e(\htf_e)\bigr)\leq\frac{\dt^2}{K}$. Each term in this  
sum is nonnegative and hence is at most $\frac{\dt^2}{K}$, and therefore we have 
$|g_e-\htf_e|\leq \dt$ for every edge $e$ (due to inverse-$K$-continuity).
Given this, we have that $l_P(g)\leq l_P(\htf)+mK\dt$ due to the $K$-Lipschitz
condition, and so $D^i(l,g)\leq D^i(l,\htf)+mK\dt$ for every commodity $i$.
Therefore, 
\begin{equation*}
\begin{split} 
\sum_e g_el_e(g_e) & \leq \sum_e g_el_e(\htf_e)+mK\dt\Bigl(\sum_i d_i\Bigr)
\leq\sum_i d_i\bigl(D^i(l,\htf)+mK\dt\bigr)+\ve \\
& \leq\sum_i d_i\bigl(D^i(l,g)+2mK\dt\bigr)+\ve 
\leq\sum_i d_i\bigl(D^i(l,g)+\e\bigr).
\end{split}
\end{equation*}

We now show that $\sum_e (g_e-\htf_e)l_e(g_e)\leq\frac{\dt^2}{K}$.
Suppose we obtain the near-optimal solution to \eqref{eqn:wardrop} by running the 
ellipsoid method with error parameter $\w=\frac{\ve}{mK\sum_i d_i}$. 
This takes time $\poly\bigl(\I,\log(\frac{1}{\w})\bigr)
=\poly\bigl(\I,\log(\frac{1}{\ve})\bigr)$.)
The near-optimality of $g$ then follows from the fact that there exists another feasible
flow $h$ satisfying:  
(i) $\|h-\htf\|_\infty\leq\w$, and so
$\Phi(h)-\Phi(\htf)\leq\sum_e(h_e-\htf_e)l_e(h_e)\leq\w m(\max_e l_e(h_e))\leq\w mK\sum_i d_i=\ve$;  
(ii) $\sum_e (h_e-g_e)l_e(g_e)=0$; 
see, e.g., Sections 3 and 4 and in particular, Lemma 4.5, in~\cite{ShmoysS06}. 
Thus, we have 
$\sum_e(\htf_e-g_e)l_e(g_e)\geq 0-\w m(\max_e l_e(g_e))\geq-\w mK\sum_i d_i\geq-\frac{\dt^2}{K}$.  
\end{proofof}


\paragraph{Definition of general nonatomic congestion games.}
This is the following generalization of nonatomic routing games. The edge set is now
replaced by a set $E$ of resources, and there are $k$ player-types. 
Each resource $e$ has a nonnegative, continuous, and strictly increasing delay function,  
$l_e:\R_+\mapsto\R_+$. Each player-type $i$ is described by a player-volume $d_i$ and an
explicitly-given non-empty strategy set $\Pc^i\sse 2^E$. 
The combined strategy-choices of the infinitely-many infinitesimal players of each type
$i$ can be described by an assignment $f=(f^1,\ldots,f^k)$, where $f^i:\Pc^i\mapsto\R_+$
satisfies $\sum_{P\in\Pc^i}f^i_P=d_i$; the cost incurred by a strategy
$Q\in\bigcup_i\Pc^i$ is then $l_Q(f):=\sum_{e\in Q}l_e(f_e)$, where
$f_e=\sum_{P\in\bigcup_i\Pc^i:e\in P}f^i_P$.  
We define $D^i(l,f)$ and an $\e$-equilibrium as before: so
$D^i(l,f)=\min_{P\in\Pc^i}l_P(f)$, and $f$ is an $\e$-equilibrium if 
$\sum_ef_el_e(f_e)\leq\sum_i d_i(D^i(l,f)+\e)$.
A Nash equilibrium or Nash assignment is a $0$-equilibrium, and 
is known to be unique.

The question with tolls is whether one can impose tolls $\tau\in\R^E$ on
resources---which, as before, yield delay functions $(l^\tau_e(x):=l_e(x)+\tau_e)_e$%
---in order to achieve a target assignment $f^*$ as the Nash assignment, or ensure that
$(f^*_e)_e$ is component-wise close to the Nash assignment.

\begin{proofof}{Theorem~\ref{thm:atomic}}
We first recall the definition of a Nash equilibrium.
A Nash equilibrium of the atomic splittable routing game is a feasible flow
$f$ such that $\sum_e f^i_el_e(f_e)\leq\sum_e g^i_el_e(f_e-f^i_e+g^i_e)$ for every
$s_i$-$t_i$ flow $g^i$ of value $d_i$. 
Equivalently, defining the marginal latency function
$\overline{l_{i,e}}(f;x):=l_e(x)+f^i_el'_e(x)$, where $l'(x)$ is the derivative of $l$,
this means that if $f^i_P>0$ for $P\in\Pc^i$, then $P$ is a shortest $s_i$-$t_i$ path
under the edge costs $\bigl(\overline{l_{i,e}}(f;f_e)\bigr)_e$. 

We use the ellipsoid method and dovetail the proof of Theorem~\ref{thm:ellipsoidnonatomic}.
Given the current ellipsoid center $(\hl,\htau)$, we obtain a separating hyperplane as in
the proof of Theorem~\ref{thm:ellipsoidnonatomic}, except that we use the marginal delay
functions $\bigl(\overline{\hl^{\htau}_{i,e}}\bigr)_{i,e}$. 
Let $g=f(l^*,\htau)=(g^i)_{i\leq k}$ be the flow returned by the oracle. If $g^i=f^{*i}$
for all $i$, then we are done, so suppose otherwise. 
Suppose that $f(\hl,\htau)\neq f^*$, that is, there is some $i$ such that
$f(\hl,\htau)^i\neq f^{*i}$. Note that this can be efficiently determined.
We can find a player $j$ and paths $P, Q\in\Pc^j$ such that $f^{*j}_P>0$, but   
$\sum_{e\in P}\overline{\hl^{\htau}_{j,e}}(f^*;f^*_e)>\sum_{e\in Q}\overline{\hl^{\htau}_{j,e}}(f^*;f^*_e)$.
Thus, the inequality 
$$
\sum_{e\in P}\overline{l_{j,e}}(f^*;f^*_e)+\tau(P)\leq\sum_{e\in Q}\overline{l_{j,e}}(f^*;f^*_e)+\tau(Q) 
$$
where both $l$ and $\tau$ are variables is violated by $(\hl,\htau)$ but satisfied by
$(l^*,\tau^*)$ since $(l^*,\tau^*)$ induce $f^*$ (by definition). Notice that the above
inequality is indeed linear in $l$ and $\tau$.

Now suppose $f(\hl,\htau)^i=f^{*i}$ for all $i$. Then, $g\neq f^*$, we can again find a player $j$ and
paths $P, Q\in\Pc^j$ such that $g^{j}_P>0$, but    
$\sum_{e\in P}\overline{\hl^{\htau}_{j,e}}(g;g_e)>\sum_{e\in Q}\overline{\hl^{\htau}_{j,e}}(g;g_e)$.
So consider the inequality
$$
\sum_{e\in P}\overline{l_{j,e}}(g;g_e)+\htau(P)\leq\sum_{e\in Q}\overline{l_{j,e}}(g;g_e)+\htau(Q) 
$$
where now only the $l_e$s are variables. This is violated by $(\hl,\htau)$ but satisfied by
$(l^*,\tau^*)$ since $g=f(l^*,\htau)$.
\end{proofof}

\section{Proofs from Section~\ref{tolls-sepa}} \label{append-tollssepa}

\subsection{Proof of Lemma~\ref{lem:tollsepagood}}

As mentioned in the proof sketch, we first show a property that is weaker than having a
discriminating pair. To this end, we define a \emph{good pair of subgraphs} (Definition~\ref{goodgraphs}) 
and first show in Lemma~\ref{helper:sepagood} that a pair of subgraphs satisfying this
weaker property always exist.

Let $(G,\{(s_i,t_i,d_i)\}_{i\in\K})$ be a multicommodity flow instance on a sepa graph. Let $\Hc$ be the collection of parallel subgraphs of  
$G$ under a given sepa decomposition tree for $G$. For any subgraph $H \in \Hc$ we define the \emph{internal nodes of $H$} as $V^{\aint}(H) := V(H) \setminus \{s_H,t_H\}$. The \emph{internal commodities of $H$} are $\K^{\aint}(H) := \{i \in \K: \{s_i, t_i\} \cap V^{\aint}(H) \neq \emptyset\}$. The \emph{external commodities of $H$} are $\K^{\aext}(H) := \{i \in \K:\mbox{ $s_H$, $t_H$ lie on some $s_i$-$t_i$ path}\}$.

Let $f=(f^i)_{i\in\K}$ and $\tf=(\tf^i)_{i\in\K}$ be two feasible multicommodity flows. Define 

\begin{align*}
\val{f_H^i} & := \sum_{e=(s_H,v)\in E(H)}f^i_e \, , \mbox{ and } ~ \val{f_H} ~ := ~ \sum_{i\in\K^{\aext}(H)} \val{f_H^i} \, .
\end{align*}

\begin{definition} \label{goodgraphs}
Given feasible flows $f$, $\tf$ in $G$, subgraphs $H_1, H_2$ are {\em $(f,\tf,\Hc)$-good} if: 
\begin{list}{(\roman{enumi})}{\usecounter{enumi} \topsep=0.5ex \itemsep=0ex}
\item the parallel-join of $H_1$ and $H_2$ is a subgraph in $\Hc$; 
\item $f_e\geq\tf_e$ for all $e\in E(H_1)$ and $f_e\leq\tf_e$ for all $e\in E(H_2)$; and 
\item $\val{f_{H_1}}>\val{\tf_{H_1}}$ and $\val{f_{H_2}}<\val{\tf_{H_2}}$.
\end{list}
\end{definition}

\begin{lemma} \label{helper:sepagood}
For any subgraph $H$ of $G$, let $\Hc'$ be the set of subgraphs of $H$ obtained by parallel joins in a given decomposition tree of $G$, and let $f$, $\tf$ be feasible multicommodity flows in $G$. Either there exists an $(f,\tf,\Hc')$-good pair of subgraphs or one of the following
must hold. 
\begin{list}{\arabic{enumi}.}{\usecounter{enumi} \topsep=0.5ex \itemsep=0ex
    \addtolength{\leftmargin}{-2ex}} 
\item If $\val{f_H}=\val{\tf_H}$ then $f_e=\tf_e$ for all $e \in E(H)$.
\item If $\val{f_H}>\val{\tf_H}$ then $f_e\geq\tf_e$ for all $e \in E(H)$.
\item If $\val{f_H}<\val{\tf_H}$ then $f_e\leq\tf_e$ for all $e \in E(H)$.
\end{list}
\end{lemma}

\begin{proof}
We proceed by induction on the size of $H$. In the base case, when $H$ is a single edge,
there is no good pair of subgraphs, but one of the three cases clearly holds.
For the induction step, suppose $H$ is the join of subgraphs $H_1$ and $H_2$. If $H$ is the parallel join of $H_1$ and $H_2$, then any external commodities of $H$ are external commodities of $H_1$ and $H_2$ as well; similarly, any external commodities of $H_1$ and $H_2$ are external commodities of $H$ as well. Hence $\val{f_H} = \val{f_{H_1}} + \val{f_{H_2}}$. Note that if $\val{f_{H_i}}>\val{\tf_{H_i}}$ and $\val{f_{H_j}}<\val{\tf_{H_j}}$ for $i\neq j$ and $i,j \in \{1,2\}$ then $H_1$ and $H_2$ form a good pair.

To verify the three cases, suppose $\val{f_H}=\val{\tf_H}$. If $\val{f_{H_i}}=\val{\tf_{H_i}}$ for $i \in \{1,2\}$, then by induction $f_e = \tf_e$ for $e \in E(H)$. Otherwise, by the induction hypothesis for $i\neq j$ and $i,j \in \{1,2\}$, $\val{f_{H_i}}>\val{\tf_{H_i}}$ and $\val{f_{H_j}}<\val{\tf_{H_j}}$ yielding a good pair. If $\val{f_H}>\val{\tf_H}$ then again, either $\val{f_{H_i}}>\val{\tf_{H_i}}$ and $\val{f_{H_j}}<\val{\tf_{H_j}}$ yielding a good pair, or $\val{f_{H_i}}>\val{\tf_{H_i}}$ and $\val{f_{H_j}}=\val{\tf_{H_j}}$. In this case, by induction, $f_e \ge \tf_e$ for all $e \in E(H)$.

Now suppose $H_1$ and $H_2$ are in series. In this case, note that $\K^{\aext}(H) \subseteq \K^{\aext}(H_i)$ for $i \in \{1,2\}$. Further, if commodity $i \in \K^{\aext}(H_1)$ but $i \not \in \K^{\aext}(H)$, then $t_i$ must be an internal node of $H_2$. Since every $s_i$-$t_i$ path contains $s_{H_1}$, and $f$, $\tf$ are feasible flows in $G$, $\val{f_{H_1}^i} = \val{\tf_{H_1}^i}$. Similarly, if commodity $i \in \K^{\aext}(H_2)$ but $i \not \in \K^{\aext}(H)$, then $s_i$ must be an internal node of $H_1$. Since every $s_i$-$t_i$ path contains $s_{H_2}$, $\val{f_{H_2}^i} = \val{\tf_{H_2}^i}$. Thus,

\begin{align*}
\val{f_{H_1}} - \val{\tf_{H_1}} & = \sum_{i \in \K^{\aext}(H_1) \cap \K^{\aext}(H)} \left(\val{f_{H_1}^i} - \val{\tf_{H_1}^i}\right) + \sum_{i \in \K^{\aext}(H_1) \setminus \K^{\aext}(H)} \left(\val{f_{H_1}^i} - \val{\tf_{H_1}^i}\right) \\
	& = \sum_{i \in \K^{\aext}(H_1) \cap \K^{\aext}(H)} \left(\val{f_{H_1}^i} - \val{\tf_{H_1}^i}\right) \\
	& = \sum_{i \in \K^{\aext}(H)}\left(\val{f_{H_1}^i} - \val{\tf_{H_1}^i}\right) ~ = \val{f_{H}} - \val{\tf_{H}} 
\end{align*}

Similarly, $\val{f_{H_2}} - \val{\tf_{H_2}}  = \val{f_{H}} - \val{\tf_{H}}$. By induction, either there is a good subgraph, or one of the three cases in the lemma must hold.
\end{proof}

\begin{proofof}{Lemma~\ref{lem:tollsepagood}} 
Since $f$ and $\tf$ are feasible multicommodity flows and $f\neq\tf$,
Lemma~\ref{helper:sepagood} implies that there is an $(f,\tf,\Hc)$-good pair of
subgraphs $H_1$, $H_2$. So
(a) $f_e\geq\tf_e$ for all $e\in E(H_1)$ and $f_e\leq\tf_e$ for all $e\in E(H_2)$, and 
(b) $\val{f_{H_1}}>\val{\tf_{H_1}}$ and $\val{f_{H_2}}<\val{\tf_{H_2}}$.
If $f_e>\tf_e$ for all $e\in E(H_1)$, then we are done. So suppose otherwise.

In the fixed decomposition tree of $G$, consider the subgraphs in the subtree rooted at
subgraph $H_1$. Let $K$ be a minimal subgraph that contains both $f_e>\tf_e$ edges and
$f_e=\tf_e$ edges; that is, every subgraph of $K$ only contains $f_e>\tf_e$ edges or
$f_e=\tf_e$ edges but not both. 
Let $K$ be the join of $K_1$ and $K_2$, where $K_1$ contains $f_e>\tf_e$ edges. 
If $K_1$, $K_2$ are in parallel, then $K_1, K_2$ is an $(f,\tf,\Hc)$-discriminating pair.

To complete the proof, we show that it cannot be that $K_1$ and $K_2$ are in series. 
Let $v$ be the node joining $K_1$ and $K_2$, so all edges with $v$ at their head lie in $E(K_1)$, and all edges with $v$ at their tail lie in $E(K_2)$.
Given a feasible multicommodity flow $h$, define
$b_v(h)=\sum_{(v,u)\in E}h_{v,u}-\sum_{(u,v)\in E}h_{u,v}$.
Observe that $b_v(h)$ is simply the node balance $\sum_{i:v=s_i}d_i-\sum_{i:v=t_i}d_i$, and is thus
independent of the multicommodity flow. 
Therefore, $b_v(f)=b_v(\tf)$. Rearranging, this gives
$\sum_{e\in E(K_1):e=(u,v)}(f_e-\tf_e)
=\sum_{e\in E(K_2):e=(v,u)}(f_e-\tf_e)$, 
which is a contradiction.
\end{proofof}


\section{Proofs from Section~\ref{tolls-linear}} \label{append-tollslinear}


\begin{proofof}{Claim~\ref{clm:nonnegativetolls}}
We assume that in $\tau$, there is a single edge $e' = (u,w)$ with negative tolls. If there are multiple such edges, simply repeating the procedure in this proof gives the required tolls $\tau'$. If $f_{e'}(\tau) = 0$, increasing the toll on this edge does no change the equilibrium flow. Hence we assume that $\ef_{e'}(l, \tau) > 0$.

Let $E^+$ be the edge set of the graph if it is acyclic; otherwise, let 
$E^+$ be the set of edges with strictly positive flow in $g = \ef(l, \tau)$. 
Since $g$ is an equilibrium flow, the set of edges $E^+$ is acyclic. Let $\sigma(v)$ be an
ordering of the vertices given by a topological sort of the graph $(V,E^+)$. Define $S =
\{v \in V: \sigma(v) \le \sigma(u)\}$, where $e'=(u,v)$ is the edge with negative
toll. Then $s \in S$ and $t \in V \setminus S$. Let $\tau'$ be the tolls obtained by
adding $-\tau_{e'}$ to every edge $e \not \in E^+$, and also to every edge $e=(x,y) \in
E^+$ across the cut $(S, V \setminus S)$. That is, 

\begin{align*}
\tau_{xy}' & = \left \{ \begin{array}{ll}
	\tau_{xy} - \tau_{e'} & \mbox{ if $x \in S$, $y \in V \setminus S$, or $(x,y) \not \in E^+$ } \\
	\tau_{xy} & \mbox{ otherwise. }
	\end{array} \right.
\end{align*}

\noindent By this procedure, the toll does not decrease on any edge and increases to zero on edge $e'$. We claim that the flow at equilibrium remains unchanged. Consider first a path $P$ with $g_P > 0$. All edges $e \in P$ are in $E^+$, and exactly one edge crosses the cut $(S, V \setminus S)$. Hence the delay on every such path increases by exactly $-\tau_{e'}$. On any other path, there is at least one edge $e \not \in E^+$, hence the delay these paths increases by at least $-\tau_{e'}$. The flow $g$ is thus a flow on shortest paths with tolls $\tau'$, and hence $g = \ef(l,\tau')$.
\end{proofof}


\begin{proofof}{Theorem~\ref{lem:linearcont}}
We first prove (iii). Let $\tau:= \mone_r \delta$. Let $\Phi$ be the potential function as defined in~\eqref{eqn:wardrop} for the delay functions in $\varGamma$, and $\Phi^\tau$ be the potential function with delay functions that include the toll $\tau$. Note that for any flow $f$, $\Phi^\tau(f) = \Phi(f) + \tau_r f_r$. Suppose for a contradiction that $f_r(\tau) > f_r(0)$. Then

\begin{align*}
\Phi^\tau(f(0)) & = \Phi(f(0)) + \tau_r f_r(0) < \Phi(f(0)) + \tau_r f_r(\tau) < \Phi(f(\tau)) + \tau_r f_r(\tau) = \Phi^\tau(f(\tau)) \, .
\end{align*}

\noindent But this is a contradiction, since $f(\tau)$ is the unique minimizer of $\Phi^\tau$.

We now prove part (iv) of the theorem. Let $\tau := -\mone_r \delta$. We first prove the lemma for the case that $\mcs(l,f(\tau)) = \mcs(l,f(0))$, and then extend it to the case when the set of shortest-path edges differ. For two flows $f$ and $g$ of the same value in $G$, the difference $h = f-g$ is a circulation and is possibly negative on some edges. If $h_{uv} > 0$ then $(u,v)$ is a forward edge, and if $h_{uv} < 0$ then $(u,v)$ is a backward. We use $E^+$ and $E^-$ for the set of forward and backward edges respectively.

We want to define a decomposition of $h$ along cycles. For this, let $D$ be the directed graph with the same vertex set as $G$, but with $(u,v) \in E(D)$ if $(u,v) \in E$ and $h_{uv} > 0$, and $(v,u) \in E(D)$ if $(u,v) \in E$ and $h_{uv} < 0$. Then $h$ defines a circulation $\tilde{h}$ in graph $D$, where $\tilde{h}_{uv} = h_{uv}$ if $(u,v)$ is a forward edge, and $\tilde{h}_{vu} = -h_{uv}$ if $(u,v)$ is a backward edge. Let $\{\tilde{h}_C\}_{C \in \mcc}$ be a decomposition of $\tilde{h}$ along directed cycles in $D$. Then for $(u,v) \in E^+$, $h_{uv} = \sum_{C: (u,v) \in C} \tilde{h}_C$, and for $(u,v) \in E^-$, $h_{uv} = - \sum_{C: (v,u) \in C} \tilde{h}_C$.

Let edge $r=(x,y)$. We will show that $(y,x)$ is in every cycle $C$. For a contradiction, suppose there exists $C' \in \mcc$ so that $(y,x) \not \in  C'$. For any edge $e \in E^+$, $f_e(\tau) > f_e(0)$, and for any edge $e \in E^-$, $f_e(\tau) < f_e(0)$. Further, since $\mcs(f(0)) = \mcs(f(\tau))$, the sum of latencies along cycle $C'$ must be zero for both flows $f(\tau)$ and $f(0)$. However,

\begin{align*}
\sum_{e \in C'} l_e(f_e(\tau)) & = \sum_{e \in E^+ \cap C'} \left(l_e(f_e(\tau)) + \tau_e\right)  - \sum_{e \in E^- \cap C'} \left(l_e(f_e(\tau)) + \tau_e\right) \\
	& > \sum_{e \in E^+ \cap C'} \left(l_e(f_e(0)) + \tau_e\right)  - \sum_{e \in E^- \cap C'} \left(l_e(f_e(0)) + \tau_e\right) \\
	& \ge \sum_{e \in E^+ \cap C'} l_e(f_e(0))  - \sum_{e \in E^- \cap C'} l_e(f_e(0))  ~ = 0 \, .
\end{align*}

\noindent where the second inequality is because $\tau_e = 0$ for $e \neq r$, and $r \not \in E^- \cap C'$. This is a contradiction, since the sum of latencies along cycle $C'$ must be zero for flow $f(\tau)$. Thus, for every cycle $C \in \mcc$, $(y,x)$ must be in $C$.

Then

\[
f_{uv}(\tau) ~=~ f_{uv}(0) + \sum_{C \in \mcc: (u,v) \in C} f_C - \sum_{C \in \mcc: (v,u) \in C} f_C \, .
\]

Since by the claim edge $r$ is a backward edge in every cycle $C \in \mcc$, $|f_r(\tau) - f_r(0)| = |\sum_{C \in \mcc} f_C|$, which is obviously an upper bound on the change in flow on any edge.

We now extend the lemma for the case where $\mcs(f(0)) \neq \mcs(f(\tau))$. In fact, we show that for any $\epsilon > 0$, $|f_r(\tau) - f_r(0)| \ge \|f(\tau) - f(0)\|_\infty - \epsilon$. Pick $\nu = \epsilon^2/(Kd 2^{2m})$, where $m$ is the number of edges. Let $a_0 = 0$. For any $a_i$ we define

\begin{align*}
b_i & = \sup \{x \in [a_i,\delta]: \mcs(f(-\mone_r x)) = \mcs(f(-\mone_r a_i)) \}\, .
\end{align*}

\noindent and $a_{i+1} = b_i + \nu$. Let $j$ be such that $\delta \in [a_j, a_{j+1}]$. By definition, either $\delta = b_j$ or $\delta \in [b_j, a_{j+1}]$. Since the number of possible sets of shortest-path edges is $2^m$, $j \le 2^m$. Also, for all $i$, by the first part of the lemma and by continuity of equilibrium flow, $|f_r(- \mone_r a_i) - f_r(- \mone_r b_i)| \ge \|f(-\mone_r a_i) - f(- \mone_r b_i)\|_\infty$. Further by Corollary~\ref{cor:indiffclose}, $\Vert f(- \mone_r b_i) - f_r(- \mone_r a_{i+1})\Vert \le \sqrt{K d\nu}$.  Summing up,

\begin{align*}
\Vert f(0) - f(-\mone_r \delta)\Vert_\infty & \le \sum_{i=0}^j \Vert f(-\mone_r a_i) - f(-\mone_r a_{i+1})\Vert_\infty \\
	& \le \sum_{i=0}^j \Vert f(-\mone_r a_i) - f(-\mone_r b_i)\Vert_\infty + \sum_{i=0}^j \Vert f(-\mone_r b_i) - f(-\mone_r a_{i+1})\Vert_\infty \\
	& \le \sum_{i=0}^j |f_r (-\mone_r a_i) - f_r(-\mone_r b_i)| + 2^m \sqrt{Kd\nu} \\
	& \le |f_r(0) - f_r(-\mone_r \delta) | + \epsilon
\end{align*}

\noindent where the last inequality follows be the monotonicity of $f_r$ as a function of the toll on edge $r$. By taking limits, $\Vert f(0) - f(\tau)\Vert_\infty \le |f_r(0) - f_r(\tau) |$.
\end{proofof}

\end{document}